\documentclass{llncs}

\usepackage{paralist} 
\usepackage{stmaryrd} 
\usepackage{amsmath}
\usepackage{amssymb}
\usepackage{epsfig} 
\usepackage{epstopdf}
\usepackage{float}
\usepackage{graphicx}
\usepackage{caption, subcaption}
\usepackage{wrapfig}
\usepackage[hyphens]{url}
\usepackage[title]{appendix}

\newcommand{\paper}[1]{}
\usepackage{tikz}
\usetikzlibrary{shapes,arrows,fit,calc,positioning,automata}

\newcommand{\nonl}{\renewcommand{\nl}{\let\nl\oldnl}}

\makeatletter
\newif\if@restonecol
\makeatother

\usepackage[ruled,vlined,linesnumbered,lined]{algorithm2e}

\usepackage{xcolor,colortbl}
\definecolor{Gray}{gray}{0.95}
\definecolor{LightCyan}{rgb}{0.88,1,1}

 \usepackage{relsize}

\newcommand{\eat}[1]{}

\renewcommand{\implies}{\Rightarrow}

\newcommand{\comment}[1]{{\color{black} \hfill \CommentSty{// #1}}}

\newcommand{\mcomment}[1]{}

\newcommand{\cba}[1]{{\Delta_{#1}}}
\newcommand{\cbb}[1]{{\Gamma_{#1}}}

\newcommand{\cbar}[1]{\delta_{#1}}
\newcommand{\cbbr}[1]{\gamma_{#1}}

\newcommand{\bfss}{\textsc{bfss}}
\newcommand{\parsyn}{\textsc{parSyn}}

\newcommand{\rsynth}{\textsc{RSynth}}
\newcommand{\cadet}{\textsc{Cadet}}
\newcommand{\unigen}{\textsc{UniGen}}
\newcommand{\abssynorig}{\textsc{AbsSynthe}}
\newcommand{\abssynsk}{\textsc{AbsSynthe-Skolem}}

\newcommand{\bX}{\ensuremath{\mathbf{X}}}
\newcommand{\bY}{\ensuremath{\mathbf{Y}}}
\newcommand{\bZ}{\ensuremath{\mathbf{Z}}}
\newcommand{\bpsi}{\ensuremath{\mathbf{\Psi}}}
\newcommand{\myP}{\ensuremath{\mathsf{P}}}
\newcommand{\NP}{\ensuremath{\mathsf{NP}}}
\newcommand{\BPP}{\ensuremath{\mathsf{BPP}}}
\newcommand{\PSIZE}{\ensuremath{\mathsf{PSIZE}}}
\newcommand{\PH}{\ensuremath{\mathsf{PH}}}
\newcommand{\Ppoly}{\ensuremath{\mathsf{P}/\mathsf{poly}}}
\newcommand{\Sigmatwop}{\ensuremath{\Sigma_2^\myP}}
\newcommand{\Pitwop}{\ensuremath{\Pi_2^\myP}}

\newcommand{\BFS}{\ensuremath{\mathsf{BFnS}}}
\newcommand{\DTIME}{\ensuremath{\mathsf{DTIME}}}
\newcommand{\DNNF}{\ensuremath{\mathsf{DNNF}}}
\newcommand{\AAA}{\ensuremath{\mathsf{A}}}

\newcommand{\lits}[1]{\ensuremath{lits({#1})}}

\usetikzlibrary{decorations.pathreplacing} 
\usetikzlibrary{patterns}

\usepackage[ruled,vlined,linesnumbered,lined]{algorithm2e}

\tikzset{
block/.style={
  draw, 
  rectangle, 
  minimum height=0cm, 
  minimum width=2cm, align=center,
  }, 
line/.style={->,>=latex'}
}
\tikzstyle{galivenode}=[circle,fill=green!50!black,thick,inner sep=1pt,minimum size=4mm]
\tikzstyle{ralivenode}=[circle,fill=red!70!black,thick,inner sep=1pt,minimum size=4mm]

\tikzstyle{alivenode}=[circle,fill=black!80,thick,inner sep=1pt,minimum size=4mm]
\tikzstyle{deadnode}=[circle,fill=black!10,thick,inner sep=0pt,minimum size=4mm]
\tikzstyle{rdnode}=[circle,draw,fill=red!10,thick,inner sep=2pt,minimum size=4mm]
\tikzstyle{grnode}=[circle,draw,fill=green!10,thick,inner sep=2pt,minimum size=4mm]

\tikzstyle{ynode}=[rectangle,draw=black,fill=yellow!10,thick,inner sep=3pt,minimum size=4mm]
\tikzstyle{gnode}=[rectangle,fill=green!10,thick,inner sep=1pt,minimum size=4mm]
\tikzstyle{bnode}=[rectangle,fill=red!10,thick,inner sep=1pt,minimum size=4mm]

\newcommand{\bone}{\ensuremath{\mathbf{1}}}

\newcommand{\bigO}{\mathcal{O}}

\newcommand{\wDNNF}{\textsf{wDNNF}}
\newcommand{\ol}[1]{\ensuremath{\overline{{#1}}}}
\newcommand{\nnf}[1]{\ensuremath{\widehat{{#1}}}}

\newcommand{\proj}[2]{\ensuremath{{#1}\!\!\downarrow\!\!{\small{#2}}}}
\newcommand{\suprt}[1]{\ensuremath{\mathsf{sup}({#1})}}

\begin{document}

\mainmatter

\setlength{\pdfpageheight}{\paperheight}
\setlength{\pdfpagewidth}{\paperwidth}

\title{What's hard about Boolean Functional Synthesis?}
\author{S. Akshay\and Supratik Chakraborty \and Shubham Goel \and \\ Sumith Kulal\and Shetal Shah}

\institute{Indian Institute of Technology Bombay, India}

\maketitle

\begin{abstract}
  Given a relational specification between Boolean inputs and outputs,
  the goal of Boolean functional synthesis is to synthesize each
  output as a function of the inputs such that the specification is
  met. In this paper, we first show that unless some hard conjectures
  in complexity theory are falsified, Boolean functional synthesis
  must generate large Skolem functions in the worst-case.
  Given this inherent hardness, what does one do to
  solve the problem?  We present a two-phase algorithm,
  where the first phase is efficient both in
  terms of time and size of synthesized functions, and solves a large
  fraction of benchmarks. To explain this surprisingly good
  performance, we provide a sufficient condition under which the first
  phase must produce correct answers. When this condition fails,
  the second phase builds upon the result of the first phase, possibly
  requiring exponential time and generating exponential-sized
  functions in the worst-case.
  Detailed experimental evaluation shows
  our algorithm to perform better than other techniques for
  a large number of benchmarks.
\end{abstract}


\section{Introduction}
\label{sec:intro}
The algorithmic synthesis of Boolean functions satisfying relational
specifications has long been of interest to logicians and computer
scientists.  Informally, given a Boolean relation between input and
outupt variables denoting the specification, our goal is to synthesize
each output as a function of the inputs such that the relational
specification is satisfied.  Such functions have also been
called \emph{Skolem functions} in the literature~\cite{fmcad2015:skolem,bierre}.  Boole~\cite{boole1847} and Lowenheim~\cite{lowenheim1910} studied variants of this problem in the context of finding most general unifiers.  While these studies are theoretically elegant, implementations of the underlying techniques
have been found to scale poorly beyond small problem
instances~\cite{macii1998}.  More recently, synthesis of Boolean
functions has found important applications in a wide range of contexts
including reactive strategy
synthesis~\cite{alur2005,syntcomp15,Zhu17}, certified QBF-SAT
solving~\cite{jiang2,Rabe15,Balabanov12,Niemetz12}, automated program
synthesis~\cite{Gulwani,SoLe13}, circuit repair and
debugging~\cite{Jo12}, disjunctive decomposition of symbolic
transition relations~\cite{Trivedi} and the like.  This has spurred
recent interest in developing practically efficient Boolean function
synthesis algorithms.  The resulting new generation of
tools~\cite{bierre,fmcad2015:skolem,tacas2017,rsynth,rsynth:fmcad2017,Rabe15,Rabe16}
have enabled synthesis of Boolean functions from much larger and more
complex relational specifications than those that could be handled by
earlier techniques, viz.~\cite{Jian,jiang2,macii1998}.  

In this paper, we re-examine the Boolean functional synthesis problem
from both theoretical and practical perspectives.  Our investigation
shows that unless some hard conjectures in complexity theory are
falsified, Boolean functional synthesis must necessarily generate
super-polynomial or even exponential-sized Skolem functions, thereby requiring
super-polynomial or exponential time, in the worst-case.  Therefore, it is unlikely
that an efficient algorithm exists for solving all instances of
Boolean functional synthesis.  There are two ways to address this
hardness in practice: (i) design algorithms that are provably
efficient but may give ``approximate'' Skolem functions that are
correct only on a fraction of all possible input assignments, or
(ii) design a phased algorithm, wherein the initial phase(s) is/are
provably efficient and solve a subset of problem instances, and
subsequent phase(s) have worst-case exponential behaviour and solve
all remaining problem instances.  In this paper, we combine the two
approaches while giving heavy emphasis on efficient instances. We
also provide a sufficient condition for our algorithm to be
efficient, which indeed is borne out by our experiments.

The primary contributions of this paper can be summarized as follows.
\begin{enumerate}
\item 
We start by showing that 
unless {\myP} = {\NP}, there exist problem instances where Boolean functional synthesis must take super-polynomial time. Also, unless the Polynomial Hierarchy collapses to the second level, there must exist problem instances where Boolean functional synthesis must generate super polynomial sized Skolem functions.  Moreover, if the non-uniform exponential time hypothesis~\cite{Chen12}
holds, there exist problem instances where Boolean functional
synthesis must generate exponential sized Skolem functions, 
thereby also requiring at least exponential time.
\item We present a new two-phase algorithm for Boolean functional synthesis.
\begin{enumerate}
\item Phase 1 of our algorithm generates candidate Skolem functions
of size polynomial in the input specification. This phase makes
polynomially many calls to an {\NP} oracle (SAT solver in practice).
Hence it directly benefits from the progess made by the SAT solving
community, and is efficient in practice.  Our experiments indicate
that Phase 1 suffices to solve a large majority of publicly available
benchmarks.
\item However, there are indeed cases where the first phase is not
enough (our theoretical results imply that such cases likely exist).
In such cases, the first phase provides good candidate Skolem functions as
starting points for the second phase. Phase 2 of our algorithm starts
from these candidate Skolem functions, and uses a CEGAR-based approach to
produce correct Skolem functions whose size may indeed be exponential in
the input specification.

\end{enumerate}
\item We analyze the surprisingly good performance of the first phase (especially in light of the
theoretical hardness results) and show a sufficient condition on the
structure of the input representation that guarantees correctness of
the first phase.  Interestingly, popular representations like ROBDDs~\cite{bryant1986}
give rise to input structures that satisfy this condition. The
goodness of Skolem functions generated in this phase of the algorithm
can also be quantified with high confidence by invoking an approximate
model counter~\cite{approxmc}, whose complexity lies in
${\BPP}^{\NP}$.
\item We conduct an extensive set of experiments over a variety of benchmarks, and show that
our algorithm performs favourably vis-a-vis state-of-the-art algorithms for Boolean functional synthesis. 
\end{enumerate}
\paragraph{\bfseries Related work}
The literature contains several early theoretical studies on variants
of Boolean functional
synthesis~\cite{boole1847,lowenheim1910,deschamps1972,boudet1989,martin1989,baader1999}. More
recently, researchers have tried to build practically efficient
synthesis tools that scale to medium or large problem
instances. In~\cite{bierre}, Skolem functions for $\bX$ are extracted
from a proof of validity of $\forall \bY \exists\bX\, F(\bX,\bY)$.
Unfortunately, this doesn't work when $\forall \bY \exists\bX\,
F(\bX,\bY)$ is not valid, despite this class of problems being
important, as discussed in~\cite{rsynth,tacas2017}. Inspired by the
spectacular effectiveness of CDCL-based SAT solvers, an incremental
determinization technique for Skolem function synthesis was proposed
in~\cite{Rabe16}.  
In~\cite{Jian,Trivedi}, a synthesis approach based on
iterated compositions was proposed.  Unfortunately, as has been noted
in~\cite{fmcad2015:skolem,rsynth}, this does not scale to large
benchmarks.  A recent work~\cite{rsynth} adapts the
composition-based approach to work with ROBDDs.  For
factored specifications, ideas from symbolic model checking using
implicitly conjoined ROBDDs have been used to enhance the
scalability of the technique further in~\cite{rsynth:fmcad2017}. 
In the genre of CEGAR-based techniques, \cite{fmcad2015:skolem} showed
how CEGAR can be used to synthesize Skolem functions from factored
specifications.  Subsequently, a compositional and parallel technique
for Skolem function synthesis from arbitrary specifications
represented using AIGs was presented in~\cite{tacas2017}. The second
phase of our algorithm builds on some of this work. 
In addition to the above techniques, template-based~\cite{Gulwani} or
sketch-based~\cite{SRBE05} approaches have been found to be effective
for synthesis when we have information about the set of candidate
solutions. A framework for functional synthesis that reasons about
some unbounded domains such as integer arithmetic, was proposed
in~\cite{KMPS10}.

\section{Notations and Problem Statement}
\label{sec:prelim}
A Boolean formula $F(z_1, \ldots z_p)$ on $p$ variables is a mapping
$F: \{0, 1\}^p \rightarrow \{0,1\}$.  The set of variables $\{z_1,
\ldots z_p\}$ is called the \emph{support} of the formula, and denoted
$\suprt{F}$.  A \emph{literal} is either a variable or its
complement. We use $F|_{z_i = 0}$ (resp. $F|_{z_i = 1}$) to denote the
positive (resp. negative) cofactor of $F$ with respect to $z_i$.
A \emph{satisfying assignment} or \emph{model} of $F$ is a mapping of
variables in $\suprt{F}$ to $\{0,1\}$ such that $F$ evaluates to $1$
under this assignment.  
If $\pi$ is a model of $F$, we write
$\pi
\models F$ and use $\pi(z_i)$ to denote the value assigned to $z_i \in
\suprt{F}$ by $\pi$.  Let $\bZ = (z_{i_1}, z_{i_2}, \ldots z_{i_j})$
be a sequence of variables in $\suprt{F}$.  We use $\proj{\pi}{\bZ}$
to denote the projection of $\pi$ on $\bZ$, i.e. the sequence
$(\pi(z_{i_1}), \pi(z_{i_2}), \ldots \pi(z_{i_j}))$.  

A Boolean formula is in \emph{negation normal form (NNF)} if (i) the
only operators used in the formula are conjunction ($\wedge$),
disjunction ($\vee$) and negation ($\neg$), and (ii) negation is
applied only to variables. Every Boolean formula can be converted to a
semantically equivalent formula in NNF.  We assume an NNF formula is
represented by a rooted directed acyclic graph (DAG), where nodes are
labeled by $\wedge$ and $\vee$, and leaves are labeled by literals.
In this paper, we use AIGs~\cite{aiger} as the initial
representation of specifications.  Given an AIG with $t$ nodes, an
equivalent NNF formula of size $\bigO(t)$ can be constructed in
$\bigO(t)$ time.  We use $|F|$ to denote the number of nodes in a DAG
represention of $F$.

Let $\alpha$ be the subformula represented by an internal node $N$
(labeled by $\wedge$ or $\vee$) in a DAG representation of an NNF
formula.  We use $\lits{\alpha}$ to denote the set of literals
labeling leaves that have a path to the node $N$ representing $\alpha$
in the AIG.
A formula is said to be in \emph{weak decomposable NNF}, or {\wDNNF},
if it is in NNF and if for every $\wedge$-labeled internal node in the
AIG, the following holds: let $\alpha
= \alpha_1 \wedge \ldots \wedge \alpha_k$ be the subformula represented by
the internal node. Then, there is no literal $l$ and distinct indices
$i, j \in \{1, \ldots k\}$ such that $l \in \lits{\alpha_i}$ and $\neg
l \in \lits{\alpha_j}$.  Note that {\wDNNF} is a weaker structural
requirement on the NNF representation vis-a-vis the well-studied
{\DNNF} representation, which has elegant
properties~\cite{darwiche-jacm}.  Specifically, every {\DNNF} formula
is also a {\wDNNF} formula.

We say a \emph{literal} $l$ is
\emph{pure} in $F$ iff the NNF representation of $F$ has a leaf labeled $l$, but no leaf
labeled $\neg l$.  $F$ is said to be \emph{positive unate} in
$z_i \in \suprt{F}$ iff $F|_{z_i = 0} \Rightarrow F|_{z_i = 1}$.
Similarly, $F$ is said to be
\emph{negative unate} in $z_i$ iff $F|_{z_i =
1} \Rightarrow F|_{z_i = 0}$.  Finally, $F$ is \emph{unate} in $z_i$
if $F$ is either positive unate or negative unate in $z_i$. A function
that is not unate in $z_i \in \suprt{F}$ is said to be \emph{binate}
in $z_i$. 

We also use $\bX = (x_1, \ldots x_n)$ to denote a sequence of Boolean
outputs, and $\bY = (y_1, \ldots y_m)$ to denote a sequence of Boolean
inputs.  The \emph{Boolean functional synthesis} problem, henceforth
denoted {\BFS}, asks: given a Boolean formula $F(\bX, \bY)$ specifying
a relation between inputs $\bY = (y_1, \ldots y_m)$ and outputs $\bX =
(x_1, \ldots x_n)$, determine functions $\bpsi =
(\psi_1(\bY), \ldots \psi_n(\bY))$ such that $F(\bpsi, \bY)$ holds
whenever $\exists \bX F(\bX, \bY)$ holds. Thus,
$\forall \bY \exists \bX\, \left(F(\bX, \bY)\right.\Leftrightarrow \left.F(\bpsi, \bY)\right)$
must be rendered valid.  The function $\psi_i$ is called
a \emph{Skolem function} for $x_i$ in $F$, and $\bpsi =
(\psi_1,\ldots \psi_n)$ is called a \emph{Skolem function vector} for
$\bX$ in $F$.

For $1 \le i \le j \le n$, let $\bX_{i}^j$ denote the subsequence
$(x_i, x_{i+1}, \ldots x_j)$ and let $F^{(i-1)}(\bX_i^n, \bY)$ denote
$\exists
\bX_1^{i-1} F(\bX_1^{i-1}, \bX_i^n, \bY)$. It has been argued
in~\cite{fmcad2015:skolem,rsynth,tacas2017,Jian} that given a
relational specification $F(\bX, \bY)$, the {\BFS}
problem can be solved by first ordering the outputs, say as
$x_1 \prec x_2 \cdots \prec x_n$, and then synthesizing a function
$\psi_i(\bX_{i+1}^n, \bY)$ for each $x_i$ such that $F^{(i-1)}(\psi_i,
\bX_{i+1}^n, \bY) \Leftrightarrow \exists x_i F^{(i-1)}(x_i,
\bX_{i+1}^n, \bY)$.  Once all such $\psi_i$ are obtained, one can
substitute $\psi_{i+1}$ through $\psi_{n}$ for $x_{i+1}$ through $x_n$
respectively, in $\psi_i$ to obtain a Skolem function for $x_i$ as a
function of only $\bY$.  We adopt this approach, and
therefore focus on obtaining $\psi_i$ in terms of $\bX_{i+1}^n$ and
$\bY$.  Furthermore, we know
from~\cite{fmcad2015:skolem,Jian} that a function $\psi_i$
is a Skolem function for $x_i$ iff it satisfies
$\cba{i}^{F} \Rightarrow \psi_i \Rightarrow \neg \cbb{i}^{F}$, where
$\cba{i}^{F}
\equiv$ $\neg \exists \bX_1^{i-1}\,$ $F(\bX_1^{i-1}, 0, \bX_{i+1}^n,
\bY)$, and $\cbb{i}^{F} \equiv$ $\neg \exists \bX_1^{i-1} F(\bX_1^{i-1},
1, \bX_{i+1}^n, \bY)$.  When $F$ is clear from the context, we often
omit it and write $\cba{i}$ and $\cbb{i}$.  It is easy to see that
both $\cba{i}$ and $\neg\cbb{i}$ serve as Skolem
functions for $x_i$ in $F$.

\section{Complexity-theoretical limits}
\label{sec:eth}
\newcommand{\Pitoop}{$\Pi^P_2$}
\newcommand{\ETH}{\textsf{ETH}}
\newcommand{\ETHnu}{\mathsf{ETH_{nu}}}
\newcommand{\EXPT}{\textsf{EXPTIME}}

In this section, we investigate the computational complexity of {\BFS}.
It is 
easy to see that {\BFS} can be solved in \EXPT. Indeed a
naive solution would be to enumerate all possible values of inputs
$\bY$ and invoke a SAT solver to find values of $\bX$ corresponding to
each valuation of $\bY$ that makes $F(\bX, \bY)$ true.  This requires
worst-case time exponential in the number of inputs and outputs, and may produce an
exponential-sized circuit.  Given this 
one can ask if we can
develop a better algorithm that works faster and synthesizes ``small''
Skolem functions in all cases? Our first result shows that existence
of such small Skolem functions would violate hard complexity-theoretic
conjectures.

\begin{theorem}
  \label{thm:hardness}
  \begin{enumerate}
  \item Unless $\myP=\NP$, there exist problem instances where any
    algorithm for {\BFS} must take super-polynomial
    time.
\item Unless $\Sigmatwop=\Pitwop$, there exist problem instances where \BFS\ must generate Skolem functions of size super-polynomial in the input size.
  \item Unless the \emph{non-uniform exponential-time hypothesis (or
    $\ETHnu$)} fails, there exist problem instances where any algorithm for {\BFS} must generate Skolem functions of size exponential in the input size.
  \end{enumerate}
\end{theorem}
The assumption in the first statement implies that the Polynomial Hierarchy (\PH) collapses completely (to level 1), while the second implies that \PH\ collapses to level 2. A consequence of the third statement is that, under this hypothesis, there must exist an instance of {\BFS} for which any algorithm must take \EXPT\ time. 

The exponential-time hypothesis
$\ETH$ and its strengthened version, the non-uniform exponential-time
hypothesis $\ETHnu$ are unproven computational hardness assumptions
(see~\cite{IP01},\cite{Chen12}), which have been used to show that
several classical decision, functional and parametrized NP-complete
problems (such as clique) are unlikely to have sub-exponential
algorithms. 
$\ETHnu$ states that there is no family of algorithms
(one for each family of inputs of size $n$) that can solve 3-SAT in
subexponential time. In~\cite{Chen12} it is shown that if $\ETHnu$
holds, then \emph{$p$-Clique, the parametrized clique problem}, cannot
be solved in sub-exponential time, i.e., for all $d\in \mathbb{N}$,
and sufficiently large fixed $k$, determining whether a graph $G$ has
a clique of size $k$ cannot be done in \DTIME $(n^d)$.

\begin{proof}
  We describe a reduction from $p$-Clique to {\BFS}.
  Given an undirected graph $G=(V,E)$ on $n$-vertices and a number $k$
  (encoded in binary), we want to check if $G$ has a clique of size
  $k$. We encode the graph as follows: each vertex $v \in V$ is
  identified by a unique number in $\{1, \ldots n\}$, and for every
  $(i, j)\in V \times V$, we introduce an input variable $y_{i,j}$
  that is set to $1$ iff $(i, j) \in E$.  We call the resulting vector
  of input variables $\vec{y}$. We also have additional input
  variables $\vec{z}=z_1,\ldots z_m$, which represent the binary
  encoding of $k$ ($m= \lceil\log_2 k\rceil$). Finally, we introduce
  output variables $x_v$ for each $v \in V$, whose values 
  determine which vertices are present in the clique.  Let
  $\vec{x}$ denote the vector of $x_v$ variables.

Given inputs $\bY=\{\vec{y},\vec{z}\}$, and outputs $\bX=\{\vec{x}\}$,
our specification is represented by a circuit $F$ over $\bX,\bY$ that
verifies whether the vertices encoded by $\bX$ indeed form a
$k$-clique of the graph $G$.  The circuit $F$ is constructed as
follows:
\begin{enumerate}
\item For every $i, j$ such that $1 \le i < j \le n$, we construct a
  sub-circuit implementing $x_i \wedge x_j \Rightarrow y_{i,j}$. The
  outputs of all such subcircuits are conjoined to give an
  intermediate output, say $\mathsf{EdgesOK}$.  Clearly, all the
  subcircuits taken together have size $\bigO(n^2)$.
\item We have a tree of binary adders implementing $x_1 + x_2 + \ldots
  x_n$.  Let the $\lceil \log_2 n \rceil$-bit output of the adder be
  denoted $\mathsf{CliqueSz}$.  The size of this adder is clearly
  $\bigO(n)$.
\item We have an equality checker that checks if $\mathsf{CliqueSz} =
  k$. Clearly, this subcircuit has size $\lceil \log_2 n \rceil$.
  Let the output of this equality checker be called $\mathsf{SizeOK}$.
\item The output of the specification circuit $F$ is $\mathsf{EdgesOK}
  \wedge \mathsf{SizeOK}$.
\end{enumerate}

Given an instance $\bY=\{\vec{y},\vec{z}\}$ of $p$-Clique, we now
consider the specification $F(\bX,\bY)$ as constructed above and feed
it as input to any algorithm {\AAA} for solving {\BFS}.  Let $\bpsi$ be
the Skolem function vector output by {\AAA}.  For each $i \in \{1,
\ldots n\}$, we now feed $\psi_i$ to the input $y_i$ of the circuit
$F$.  This effectively constructs a circuit for $F(\bpsi,\bY)$.  It is
easy to see from the definition of Skolem functions that for every
valuation of $\bY$, the function $F(\bpsi, \bY)$ evaluates to $1$ iff
the graph encoded by $\bY$ contains a clique of size $k$.

Using this reduction, we can complete the proofs of our statements:
\begin{enumerate}
\item If the circuits for the Skolem functions $\bpsi$ are super-polynomial size, then of course any algorithm generating $\bpsi$ must take super-polynomial time. On the other hand, if the circuits for the Skolem functions $\bpsi$ are always poly-sized, then $F(\bpsi, \bY)$ is polynomial-sized, and evaluating
  it takes time that is polynomial in the input size.  Thus, if {\AAA} is a polynomial-time algorithm, we also get an algorithm for solving $p$-Clique in polynomial time, which  implies that $\myP=\NP$.
\item  If the circuits for the Skolem functions $\bpsi$ produced by algorithm {\AAA} are always polynomial-sized, then $F(\bpsi, \bY)$ is polynomial-sized. Thus,with polynomial-sized circuits we are able to solve $p$-Clique. Recall that problems that can be solved using polynomial-sized circuits are said to be in the class \PSIZE\ (equivalently called \Ppoly). But since $p$-Clique is an \NP-complete problem, we obtain that $\NP\subseteq \Ppoly$. By the Karp-Lipton Theorem~\cite{KL82}, this implies that $\Sigmatwop=\Pitwop$, which implies that \PH\ collapses to level 2. 
\item If the circuits for the Skolem functions $\bpsi$ are
  sub-exponential sized in the input $n$, then $F(\bpsi, \bY)$ is also
  sub-exponential sized and can be evaluated in sub-exponential time.
  It then follows that we can solve any instance p-Clique of input
  length $n$ in sub-exponential time -- a violation of $\ETHnu$. Note
  that since our circuits can change for different input lengths, we
  may have different algorithms for different $n$. Hence we have to
  appeal to the non-uniform variant of $\ETH$.\qed
\end{enumerate}
\end{proof}

Theorem~\ref{thm:hardness} implies that efficient algorithms for
{\BFS} are unlikely.  We therefore propose a two-phase algorithm to
solve {\BFS} in practice.  The first phase runs in polynomial time
relative to an {\NP}-oracle and generates polynomial-sized
``approximate'' Skolem functions.  We show that under certain
structural restrictions on the NNF representation of $F$, the first
phase always returns exact Skolem functions.  However, these
structural restrictions may not always be met.  An {\NP}-oracle can be
used to check if the functions computed by the first phase are indeed
exact Skolem functions.  In case they aren't, we proceed to the second
phase of our algorithm that runs in worst-case exponential time.
Below, we discuss the first phase in detail.  The second phase is an
adaptation of an existing CEGAR-based technique and is described
briefly later.

\section{Phase 1: Efficient polynomial-sized synthesis}
\label{sec:nnf}

An easy consequence of the definition of unateness is the following.
\begin{proposition}
  \label{prop:unate}
If $F(\bX,\bY)$ is positive (resp. negative) unate in $x_i$, then
$\psi_i = 1$ (resp. $\psi_i = 0$) is a correct Skolem function for
$x_i$.
\end{proposition}
\begin{proof}
Recall that $F$ is positive unate in $x_i$ means that  $F|_{x_i=0}\implies F|_{x_i=1}$. We start by observing that $\exists x_i F=F|_{x_i=0}\vee F|_{x_i=1}\implies F|_{x_i=1}$. Conversely, $F|_{x_i=1} \implies \exists x_i F$. Hence, we conclude that $1$ is indeed a correct Skolem function for $x_i$ in $F$. The proof for negative unateness follows on the same lines.\qed
\end{proof}
The above result gives us a way to identify outputs $x_i$ for which a Skolem function can be easily computed. Note
that if $x_i$ (resp. $\neg x_i$) is a pure literal in $F$, then $F$ is
positive (resp. negative) unate in $x_i$.  However, the converse is
not necessarily true.  In general, a semantic check is necessary for
unateness. In fact, it follows from the definition of unateness that
$F$ is positive (resp. negative) unate in $x_i$, iff the formula
$\eta_i^+$ (resp. $\eta_i^-$) defined below is unsatisfiable.
\begin{align}
  \eta_i^+ &= F(\bX_{1}^{i-1},0,\bX_{i+1}^n, \bY) \wedge \neg
  F(\bX_{1}^{i-1},1,\bX_{i+1}^n, \bY). \label{eqn:eta_i_plus}\\
  \eta_i^- &= F(\bX_{1}^{i-1},1,\bX_{i+1}^n, \bY) \wedge \neg
  F(\bX_{1}^{i-1},0,\bX_{i+1}^n, \bY). \label{eqn:eta_i_minus}
\end{align}
Note that each such check involves a single invocation of an NP-oracle,
and a variant of this method is described in~\cite{ABCH02}.

If $F$ is binate in an output $x_i$, Proposition~\ref{prop:unate}
doesn't help in synthesizing $\psi_i$.  Towards synthesizing Skolem
functions for such outputs, recall the definitions of $\cba{i}$ and
$\cbb{i}$ from Section~\ref{sec:prelim}.  Clearly, if we can compute
these functions, we can solve {\BFS}.  While computing $\cba{i}$ and
$\cbb{i}$ \emph{exactly} for all $x_i$ is unlikely to be efficient in
general (in light of Theorem~\ref{thm:hardness}), we show that
polynomial-sized ``good'' approximations of $\cba{i}$ and $\cbb{i}$
can be computed efficiently. As our experiments show, these approximations are good enough to solve {\BFS} for several benchmarks. Further, with an access to an {\NP}-oracle, we can also check when these approximations are indeed good enough.

Given a relational specification $F(\bX, \bY)$, we use $\nnf{F}(\bX,
\ol{\bX}, \bY)$ to denote the formula obtained by first converting $F$
to NNF, and then replacing every occurrence of $\neg x_i ~(x_i \in
\bX)$ in the NNF formula with a fresh variable $\ol{x_i}$.  As an
example, suppose $F(\bX, \bY) = (x_1 \vee \neg (x_2 \vee y_1)) \vee
\neg(x_2 \vee \neg(y_2 \wedge \neg y_1))$.  Then $\nnf{F}(\bX,
\ol{\bX}, \bY) = (x_1 \vee (\ol{x_2} \wedge \neg y_1)) \vee (\ol{x_2}
\wedge y_2 \wedge \neg y_1)$. Then, we have 
\begin{proposition}\label{prop:f-fnnf}
  \begin{enumerate}
  \item[(a)] $\nnf{F}(\bX, \ol{\bX}, \bY)$ is positive unate 
    in both $\bX$ and $\ol{\bX}$.
  \item[(b)] Let $\neg \bX$ denote $(\neg x_1, \ldots \neg x_n)$.
    Then $F(\bX, \bY) \Leftrightarrow \nnf{F}(\bX, \neg \bX, \bY)$.
  \end{enumerate}
\end{proposition}
For every $i \in \{1, \ldots n\}$, we can split $\bX = (x_1, \ldots
x_n)$ into two parts, $\bX_{1}^i$ and $\bX_{i+1}^n$, and represent
$\nnf{F}(\bX, \ol{\bX}, \bY)$ as $\nnf{F}(\bX_{1}^i,\bX_{i+1}^n,
\ol{\bX}_{1}^i, \ol{\bX}_{i+1}^n, \bY)$.
We use these representations of $\nnf{F}$ interchangeably, depending
on the context.  For $b, c \in \{0,1\}$, let $\mathbf{b}^{i}$
(resp. $\mathbf{c}^{i}$) denote a vector of $i$ $b$'s (resp. $c$'s).
For notational convenience, we use
$\nnf{F}(\mathbf{b}^i, \bX_{i+1}^n, \mathbf{c}^i, \ol{\bX}_{i+1}^n, \bY)$
to denote
$\nnf{F}(\bX_{1}^i, \bX_{i+1}^n, \ol{\bX}_{1}^i, \ol{\bX}_{i+1}^n, \bY)|_{\bX_{1}^i
= \mathbf{b}^i, \ol{\bX}_{1}^i = \mathbf{c}^i}$ in the subsequent
discussion.  The following is an easy consequence
of Proposition~\ref{prop:f-fnnf}.
\begin{proposition}\label{prop:exists-bounds}
  For every $i \in \{1, \ldots n\}$, the following holds:\\
   $\nnf{F}(\mathbf{0}^i, \bX_{i+1}^n, \mathbf{0}^i, \neg{\bX}_{i+1}^n, \bY)
   ~\Rightarrow~ \exists \bX_{1}^i F(\bX, \bY)
   ~\Rightarrow~ \nnf{F}(\mathbf{1}^i, \bX_{i+1}^n, \mathbf{1}^i, \neg{\bX}_{i+1}^n, \bY)$
\end{proposition}
Proposition~\ref{prop:exists-bounds} allows us to bound $\cba{i}$ and
$\cbb{i}$ as follows.
\begin{lemma}\label{lemma:delta-gamma-bounds}
  For every $x_i \in \bX$, we have:
  \begin{enumerate}
    \item[(a)] $\neg\nnf{F}(\mathbf{1}^{i-1}0, \bX_{i+1}^n,
      \mathbf{1}^i, \neg{\bX}_{i+1}^n, \bY) \Rightarrow \cba{i}
      \Rightarrow \neg\nnf{F}(\mathbf{0}^{i}, \bX_{i+1}^n,
      \mathbf{0}^{i-1}1, \neg{\bX}_{i+1}^n, \bY)$
    \item[(b)] $\neg\nnf{F}(\mathbf{1}^{i}, \bX_{i+1}^n,
      \mathbf{1}^{i-1}0, \neg{\bX}_{i+1}^n, \bY) \Rightarrow \cbb{i}
      \Rightarrow \neg\nnf{F}(\mathbf{0}^{i-1}1, \bX_{i+1}^n,
      \mathbf{0}^i, \neg{\bX}_{i+1}^n, \bY)$
  \end{enumerate}
\end{lemma}
In the remainder of the paper, we only use under-approximations of
$\cba{i}$ and $\cbb{i}$, and use $\cbar{i}$ and $\cbbr{i}$
respectively, to denote them.
Recall from Section~\ref{sec:prelim} that both $\cba{i}$ and
$\neg\cbb{i}$ suffice as Skolem functions for $x_i$.
Therefore, we propose to use either
$\cbar{i}$ or $\neg \cbbr{i}$ (depending on which has a smaller AIG)
obtained from Lemma~\ref{lemma:delta-gamma-bounds} as
our approximation of $\psi_i$. Specifically, 
\begin{eqnarray}
  \cbar{i}
      &=& \neg\nnf{F}(\mathbf{1}^{i-1}0, \bX_{i+1}^n, \mathbf{1}^i, \neg{\bX}_{i+1}^n, \bY),
      ~ \cbbr{i}
      = \neg\nnf{F}(\mathbf{1}^{i}, \bX_{i+1}^n, \mathbf{1}^{i-1}0, \neg{\bX}_{i+1}^n, \bY) \nonumber\\
      \psi_{i}
      &=& \cbar{i} \mbox{ or } \neg \cbbr{i}, \mbox{
      depending on which has a smaller AIG}\label{eq:init_est}
\end{eqnarray}

\begin{example}\label{ex:rabe-seshia}
  Consider the specification $\bX = \bY$, expressed in NNF as $F(\bX, \bY)
  \equiv$ $\bigwedge_{i=1}^n \left((x_i \wedge y_i) \vee (\neg x_i
  \wedge \neg y_i)\right)$.  As noted in~\cite{Rabe16}, this
  is a difficult example for CEGAR-based QBF solvers,
  when $n$ is large.

  From Eqn~\ref{eq:init_est}, $\cbar{i}$ $=$ $\neg(\neg y_i
  ~\wedge~ \bigwedge_{j=i+1}^n (x_j \Leftrightarrow y_j))$ $=$ $y_i
  ~\vee~ \bigvee_{j=i+1}^n (x_j \Leftrightarrow \neg y_j)$, and
  $\cbbr{i}$ $=$ $\neg(y_i ~\wedge~$ $\bigwedge_{j=i+1}^n
  (x_j \Leftrightarrow y_j))$ $=$ $\neg y_i ~\vee~ \bigvee_{j=i+1}^n
  (x_j \Leftrightarrow \neg y_j)$.  With $\cbar{i}$ as the
  choice of $\psi_{i}$, we obtain $\psi_{i}$ $=$
  $y_i ~\vee~ \bigvee_{j=i+1}^n (x_j \Leftrightarrow \neg
  y_j)$. Clearly, $\psi_{n} = y_n$. On reverse-substituting, we
  get $\psi_{n-1} = y_{n-1} ~\vee~ (\psi_n \Leftrightarrow \neg
  y_n) = y_{n-1} \vee 0 = y_{n-1}$.  Continuing in this way, we get
  $\psi_{i} = y_i$ for all $i \in \{1, \ldots n\}$.  The same
  result is obtained regardless of whether we choose $\cbar{i}$
  or $\neg\cbbr{i}$ for each $\psi_{i}$.  Thus, our 
  approximation is good enough to solve this problem. In
  fact, it can be shown that $\cbar{i} = \cba{i}$ and
  $\cbbr{i} = \cbb{i}$ for all $i \in \{1, \ldots n\}$ in this
  example. \qed
\end{example}

Note that the approximations of Skolem functions, as given in
Eqn~(\ref{eq:init_est}), are efficiently computable for all $i \in
\{1, \ldots n\}$, as they involve evaluating $\nnf{F}$ with a subset
of inputs set to constants.  This takes no more than $\bigO(|F|)$ time
and space.  As illustrated by Example~\ref{ex:rabe-seshia}, these
approximations also often suffice to solve {\BFS}.  The following
theorem partially explains this.

\begin{theorem}\label{lemma:init_sk_good}
  \begin{enumerate}
  \item[(a)] For $i \in \{1, \ldots n\}$, suppose the following holds:
    \begin{align}
    \forall j \in \{1, \ldots i\}\;~ \nnf{F}(\bone^j, \bX_{j+1}^n, \bone^j, \ol{\bX}_{j+1}^n, \bY) \Rightarrow &
    \nnf{F}(\bone^{j-1}0, \bX_{j+1}^n, \bone^{j-1}1, \ol{\bX}_{j+1}^n, \bY) \nonumber\\
    & \hspace*{-5mm}~\vee\ 
    \nnf{F}(\bone^{j-1}1, \bX_{j+1}^n, \bone^{j-1}0, \ol{\bX}_{j+1}^n, \bY) \nonumber
    \end{align}
    Then $\exists \bX_{1}^i F(\bX, \bY) \Leftrightarrow \nnf{F}(\bone^i, \bX_{i+1}^n, \bone^i, \neg \bX_{i+1}^n,\bY)$.
  \item[(b)] If $\nnf{F}(\bX, \neg \bX, \bY)$ is in {\wDNNF}, then $\delta_i = \Delta_i$
    and $\gamma_i = \Gamma_i$ for every $i \in \{1, \ldots n\}$.
  \end{enumerate}
\end{theorem}
\begin{proof}
To prove part (a), we use induction on $i$.  The base case corresponds
to $i=1$.  Recall that $\exists \bX_1^1 F(\bX,\bY) \Leftrightarrow
\nnf{F}(1,\bX_{2}^n,0,\neg{\bX}_2^n,\bY) \vee
F(0,\bX_{2}^n,1,\neg{\bX}_2^n,\bY)$ by definition.
Proposition~\ref{prop:exists-bounds} already asserts that $\exists
\bX_1^1 F(\bX,\bY) \Rightarrow
\nnf{F}(1,\bX_{2}^n,1,\neg{\bX}_2^n,\bY)$.  Therefore, if the
condition in Theorem~\ref{lemma:init_sk_good}(a) holds for $i=1$, we
then have $\nnf{F}(1,\bX_{2}^n,1,\neg{\bX}_2^n,\bY) \Leftrightarrow
\nnf{F}(1,\bX_{2}^n,0,\neg{\bX}_2^n,\bY) \vee
F(0,\bX_{2}^n,1,\neg{\bX}_2^n,\bY)$, which in turn is equivalent to
$\exists \bX_1^1 F(\bX,\bY)$. This proves the base case.

Let us now assume (inductive hypothesis) that the statement of
Theorem~\ref{lemma:init_sk_good}(a) holds for $1 \le i < n$.  We prove
below that the same statement holds for $i+1$ as well.  Clearly,
$\exists \bX_1^{i+1} F(\bX,\bY) \Leftrightarrow \exists x_{i+1}
\left(\exists \bX_1^i F(\bX,\bY)\right)$.  By the inductive
hypothesis, this is equivalent to $\exists x_{i+1} \nnf{F}(\bone^i,
\bX_{i+1}^n, \bone^i, \neg \bX_{i+1}^n,\bY)$.  By definition of
existential quantification, this is equivalent to
$\nnf{F}(\bone^{i+1}, \bX_{i+2}^n, \bone^i0, \neg \bX_{i+2}^n,\bY)
\vee \nnf{F}(\bone^i0, \bX_{i+2}^n, \bone^{i+1}, \neg
\bX_{i+2}^n,\bY)$.  From the condition in
Theorem~\ref{lemma:init_sk_good}(a), we also have
\begin{align*}
  \nnf{F}(\bone^{i+1}, \bX_{i+2}^n, \bone^{i+1}, \ol{\bX}_{i+2}^n, \bY)
& \Rightarrow \nnf{F}(\bone^{i}0, \bX_{i+2}^n, \bone^{i+1},
\ol{\bX}_{i+2}^n, \bY)\\ &\vee \nnf{F}(\bone^{i+1}, \bX_{i+2}^n,
\bone^{i}0, \ol{\bX}_{i+2}^n, \bY)\end{align*}  The implication in the reverse direction follows from Proposition~\ref{prop:f-fnnf}(a).  
Thus we have a bi-implication above, which we have already seen is equivalent to $\exists \bX_1^{i+1} F(\bX,\bY)$.  This proves the inductive case.

To prove part (b), we first show that if $\nnf{F}(\bX, \neg\bX, \bY)$ is
in {\wDNNF}, then the condition in Theorem~\ref{lemma:init_sk_good}(a)
must hold for all $j \in \{1, \ldots n\}$.
  Theorem~\ref{lemma:init_sk_good}(b) then follows from the
  definitions of $\cba{i}$ and $\cbb{i}$ (see
  Section~\ref{sec:prelim}), from the statement of
  Theorem~\ref{lemma:init_sk_good}(a) and from the definitions of
  $\cbar{i}$ and $\cbbr{i}$ (see Eqn~\ref{eq:init_est}).

For $1 \leq j \leq n$, let $\zeta(\bX_{j+1}^n, \ol{\bX}_{j+1}^n,
\bY)$ denote $\nnf{F}(\bone^j, \bX_{j+1}^n, \bone^j,
\ol{\bX}_{j+1}^n, \bY)$ $\wedge$ $\neg\left(\nnf{F}(\bone^{j-1}0,
\bX_{j+1}^n, \bone^{j-1}1, \ol{\bX}_{j+1}^n, \bY) \vee \nonumber
\nnf{F}(\bone^{j-1}1, \bX_{j+1}^n, \bone^{j-1}0, \ol{\bX}_{j+1}^n,
\bY)\right)$.  To prove by contradiction, suppose $\nnf{F}$ is in {\wDNNF} but
there exists $j~(1 \le j \le n)$ such that $\zeta(\bX_{j+1}^n,
\ol{\bX}_{j+1}^n, \bY)$ is satisfiable.  Let $\bX_{j+1}^n =
\mathbf{\sigma}$, $\ol{\bX}_{j+1}^n = \mathbf{\kappa}$ and $\bY =
\mathbf{\theta}$ be a satisfying assignment of $\zeta$.  We now
consider the simplified circuit obtained by substituting $\bone^{j-1}$
for $\bX_{1}^{j-1}$ as well as for $\ol{\bX}_{1}^{j-1}$,
$\mathbf{\sigma}$ for $\bX_{j+1}^n$, $\mathbf{\kappa}$ for
$\ol{\bX}_{j+1}^n$ and $\mathbf{\theta}$ for $\bY$ in the AIG for
$\nnf{F}$.  This simplification replaces the output of every internal
node with a constant ($0$ or $1$), if the node evaluates to a constant
under the above assignment.  Note that the resulting circuit can have
only $x_j$ and $\ol{x}_j$ as its inputs.  Furthermore, since the
assignment satisfies $\zeta$, it follows that the simplified circuit
evaluates to $1$ if both $x_j$ and $\ol{x}_j$ are set to $1$, and it
evaluates to $0$ if any one of $x_j$ or $\ol{x}_j$ is set to $0$.
This can only happen if there is a node labeled $\wedge$ in the AIG
representing $\nnf{F}$ with a path leading from the leaf labeled
$x_j$, and another path leading from the leaf labeled $\neg x_j$.
This is a contradiction, since $\nnf{F}$ is in {\wDNNF}.  Therefore,
there is no $j \in \{1, \ldots n\}$ such that the condition of Theorem~\ref{lemma:init_sk_good}(a)
is violated.\qed

\end{proof}

In general, the candidate Skolem functions generated from the
approximations discussed above may not always be correct. 
Indeed, the conditions discussed above
are only sufficient, but not necessary, for the approximations
to be exact.
Hence, we need a separate check to see if our candidate
Skolem functions are correct. To do this, we 
use an \emph{error formula}
$\varepsilon_\bpsi(\bX', \bX, \bY) \equiv F(\bX', \bY) \wedge
\bigwedge_{i=1}^n (x_i \leftrightarrow \psi_i) \wedge \neg F(\bX,
\bY)$, as described in~\cite{fmcad2015:skolem}, and check its
satisfiability. The correctness of this check depends on the following
result from~\cite{fmcad2015:skolem}.
\begin{theorem}[\cite{fmcad2015:skolem}]\label{thm:epsilon-correct}
  $\varepsilon_\bpsi$ is unsatisfiable iff $\bpsi$ is a correct Skolem function vector.  
\end{theorem}

\begin{algorithm}[h!]
  \KwIn{$\nnf{F}(\bX,\bY)$ in NNF (or wDNNF) with inputs $|\bY| = m$, outputs $|\bX| = n$,}
\KwOut{ 
Candidate Skolem Functions $\bpsi=(\psi_1,\ldots, \psi_n)$}
{\bfseries Initialize}: Fix sets $U_0=U_1=\emptyset$\;
\Repeat{F is unchanged\comment{No Unate variables remaining}}{ \comment{Repeatedly check for Unate variables}\\

\For{each $x_i\in \bX\setminus (U_0\cup U_1)$}{
      \If{$\nnf{F}$ is positive unate in $x_i$ \comment{check $x_i$ pure or $\eta_i^+$ (Eqn~\ref{eqn:eta_i_plus})
      SAT} \;}{
       $\nnf{F}:=\nnf{F}[x_i=1]$, $U_1=U_1\cup \{x_i\}$}
      \ElseIf{$\nnf{F}$ is negative unate in $x_i$ \comment{$\neg x_i$ pure or $\eta^-$ (Eqn~\ref{eqn:eta_i_minus}) SAT}
      \;}{
      $\nnf{F}:=\nnf{F}[x_i=0]$, $U_0=U_0\cup\{x_i\}$}
      }}
Choose an ordering $\preceq$ of $\bX$ \comment{Section~\ref{sec:expt} discusses ordering used}\;
\For{each $x_i\in \bX$ in $\preceq$ order}{
          \If{$x_i\in U_j$ for $j\in \{0,1\}$  \comment{Assume $x_1 \preceq x_2 \preceq \ldots x_n$}\;}{$\psi_i=j$}
          \Else{$\psi_i$ is as defined in (Eq~\ref{eq:init_est})}}
\If{error formula $\epsilon_\bpsi$ is UNSAT}{terminate and output $\bpsi$}
\Else{call \textsf{Phase 2}}
\caption{\bfss}
\label{alg:easy}
\end{algorithm}

We now combine all the above ingredients to come up with
algorithm {\bfss} (for \emph{Blazingly Fast Skolem Synthesis}),
as shown in Algorithm~\ref{alg:easy}. The algorithm can be divided
into three parts. In the first part (lines 2-11), unateness is
checked. This is done in two ways: (i) we identify pure literals in
$F$ 
by simply examining the labels of leaves in the DAG representation of
$F$ in NNF, and (ii) we check the satisfiability of the formulas
$\eta_i^+$ and $\eta_i^-$, as defined in Eqn~\ref{eqn:eta_i_plus} and
Eqn~\ref{eqn:eta_i_minus}.  This requires invoking a SAT solver in the
worst-case, and is repeated at most $\mathcal{O}(n^2)$ times until
there are no more unate variables. Hence this 
requires $\mathcal{O}(n^2)$ calls to a SAT solver. Once we
have done this, by Proposition~\ref{prop:unate}, the
constants $1$ or $0$ (for positive or negative unate variables
respectively) are correct Skolem functions for these variables.

In the second part, we fix an ordering of the remaining output
variables according to an experimentally sound heuristic, as described
in Section~\ref{sec:expt}, and compute candidate Skolem functions for
these variables according to Equation~\ref{eq:init_est}.
We then check the satisfiability of the error formula $\epsilon_\bpsi$
to determine if the candidate Skolem functions are indeed correct.
If the error formula is found to be unsatisfiable, we know from
Theorem~\ref{thm:epsilon-correct} that we have the correct Skolem
functions, which can therefore be output.  This concludes phase 1 of
algorithm {\bfss}.  If the error formula is found to be satisfiable,
we move to phase 2 of algorithm {\bfss} -- an adaptation of the
CEGAR-based technique described in~\cite{fmcad2015:skolem}, and
discussed briefly in Section~\ref{sec:cegar}. It is not difficult to
see that the running time of phase 1 \t is polynomial in the size
of the input, relative to an {\NP}-oracle (SAT solver in
practice). This also implies that the Skolem functions generated can
be of at most polynomial size. Finally, from
Theorem~\ref{lemma:init_sk_good}(b) we also obtain that if $F$ is
in \wDNNF, Skolem functions generated in phase 1 are correct. From
the above reasoning, we obtain the following properties of phase 1 of
{\bfss}:
\begin{theorem}
\begin{enumerate}
\item For all unate variables, phase 1 of {\bfss} computes correct Skolem functions.
\item If $\hat{F}$ is in \wDNNF, phase 1 of {\bfss} computes all Skolem functions correctly.
\item The running time of phase 1 of {\bfss}  is polynomial in input size, relative to an {\NP}-oracle. Specifically, the algorithm makes $\mathcal{O}(n^2)$ calls to an
{\NP}-oracle.
\item The candidate Skolem functions output by phase 1 of {\bfss}
have size at most polynomial in the size of the input.
\end{enumerate}
\end{theorem}
        
\paragraph{\bfseries Discussion:}We make two crucial and related observations. First, by our hardness
results in Section~\ref{sec:eth}, we know that the above algorithm
cannot solve \BFS\ for all inputs,
unless some well-regarded complexity-theoretic conjectures fail.  As a
result, we must go to phase 2 on at least some inputs. Surprisingly,
our experiments show that this is not necessary in the majority
of benchmarks.

The second observation tries to understand why phase 1 works in most
cases in practice. While a conclusive explanation isn't easy, we
believe Theorem~\ref{lemma:init_sk_good} explains the success of phase
1 in several cases.  By~\cite{darwiche-jacm}, we know that all Boolean
functions have a {\DNNF} (and hence \wDNNF) representation, although
it may take exponential time to compute this representation. This
allows us to define two preprocessing procedures. In the first, we
identify cases where we can directly convert to \wDNNF and use the
Phase 1 algorithm above. And in the second, we use several
optimization scripts available in the ABC~\cite{abc-tool} library to
optimize the AIG representation of $\nnf{F}$.  For a majority of
benchmarks, this appears to yield a representation of $\nnf{F}$ that
allows the proof of Theorem~\ref{lemma:init_sk_good}(a) to go
through. For the rest, we apply the Phase 2 algorithm as described
below.

\paragraph{\bfseries Quantitative guarantees of ``goodness''}
Given our theoretical and practical insights of the applicability of
phase 1 of {\bfss}, it would be interesting to measure how much
progress we have made in phase 1, even if it does not give the correct
Skolem functions. One way to measure this ``goodness'' is to estimate
the number of counterexamples as a fraction of the size of the input
space.
Specifically, given the error formula, we get an
approximate count of the number of models for this
formula \emph{projected on the inputs $\bY$}.  This can be obtained
efficiently in practice with high confidence using state-of-the-art
approximate model counters, viz.~\cite{approxmc}, with complexity in
${\BPP}^{\NP}$. The approximate count thus obtained, when divided by
$2^{|\bY|}$ gives the fraction of input combinations for which the
candidate Skolem functions output by phase 1 do not work correctly.
We call this the \emph{goodness ratio} of our approximation.

\section{Phase 2: Counterexample-guided refinement}
\label{sec:cegar}
For phase 2, we can use any off-the-shelf worst-case exponential-time
Skolem function generator. However, given that we already have
candidate Skolem functions with guarantees on their ``goodness'', it
is natural to use them as starting points for phase 2. Hence, we start
off with candidate Skolem functions for all $x_i$ as computed in phase
1, and then update (or refine) them in a counterexample-driven manner.
Intuitively, a counterexample is a value of the inputs $\bY$ for which
there exists a value of $\bX$ that renders $F(\bX, \bY)$ true, but for
which $F(\bpsi, \bY)$ evaluates to false.  As shown
in~\cite{fmcad2015:skolem}, given a candidate Skolem function vector,
every satisfying assignment of the error formula $\varepsilon_{\bpsi}$
gives a counterexample.  The refinement step uses this satisfying
assignment to update an appropriate subset of the approximate
$\cbar{i}$ and $\cbbr{i}$ functions computed in phase 1. The entire
process is then repeated until no counterexamples can be found.  The
final updated vector of Skolem functions then gives a solution of the
{\BFS} problem. Note that this idea is not
new~\cite{fmcad2015:skolem,tacas2017}. The only significant enhancement
we do over the algorithm in~\cite{fmcad2015:skolem} is to use an
almost-uniform sampler~\cite{unigen2} to efficiently sample the space
of counterexamples almost uniformly.  This allows us to do refinement
with a diverse set of counterexamples, instead of using
counterexamples in a corner of the solution space of
$\varepsilon_{\bpsi}$ that the SAT solver heuristics zoom down on.

\section{Experimental results} 
\label{sec:expt}

\noindent{\bf Experimental methodology.}
Our implementation consists of two parallel pipelines that accept the
same input specification but represent them in two different ways.
The first pipeline takes the input formula as an AIG and builds an NNF
(not necessarily wDNNF) DAG, while the second pipeline builds an ROBDD
from the input AIG using dynamic variable reordering (no restrictions
on variable order), and then obtains a wDNNF representation from it
using the linear-time algorithm described in~\cite{darwiche-jacm}.
Once the NNF/wDNNF representation is built, we use Algorithm
~\ref{alg:easy} in Phase 1 and CEGAR-based synthesis using
$\unigen$\cite{unigen2} to sample counterexamples in Phase 2. We call
this ensemble of two pipelines as $\bfss$.  We compare $\bfss$ with
the following algorithms/tools: $(i)$ $\parsyn$~\cite{tacas2017},
$(ii)$ $\cadet$~\cite{cadet}, $(iii)$
$\rsynth$~\cite{rsynth:fmcad2017}, and $(iv)$ $\abssynsk$ (based on
the {\BFS} step of $\abssynorig$~\cite{abssynthe}).

Our implementation of $\bfss$ uses the ABC~\cite{abc-tool} library to
represent and manipulate Boolean functions. Two different SAT solvers
can be used with {\bfss}: ABC's default SAT solver, or
$\unigen$ \cite{unigen2} (to give almost-uniformly distributed
counterexamples). All our experiments use $\unigen$.

We consider a total of $504$ benchmarks, taken from four different
domains:
\begin{itemize}
\item[(a)] forty-eight {\em Arithmetic benchmarks}
from~\cite{rsynth}, with varying bit-widths (viz. $32$, $64$, $128$,
$256$, $512$ and $1024$) of arithmetic operators,
\item[(b)] sixty-eight {\em
Disjunctive Decomposition benchmarks} from~\cite{tacas2017}, generated
by considering some of the larger sequential circuits in the HWMCC10
benchmark suite,
\item[(c)] five {\em Factorization benchmarks}, also
from~\cite{tacas2017}, representing factorization of numbers of
different bit-widths ($8$, $10$, $12$, $14$, $16$), and
\item[(d)] three
hundred and eighty three {\em QBFEval benchmarks}, taken from the
Prenex 2QBF track of QBFEval 2017 \cite{qbfeval2017}\footnote{The
track contains $384$ benchmarks, but we were unsuccessful in converting
$1$ benchmark to some of the formats required by the various tools.}.
\end{itemize}
Since different tools accept benchmarks in different formats, each
benchmark was converted to both \texttt{qdimacs}
and \texttt{verilog/aiger} formats.  All benchmarks and the procedure
by which we generated (and converted) them are detailed
in \cite{cav2018:benchmarks}. Recall that we use two pipelines for
{\bfss}.  We use ``balance; rewrite -l; refactor -l; balance; rewrite
-l; rewrite -lz; balance; refactor -lz; rewrite -lz; balance'' as the
ABC script for optimizing the AIG representation of the input
specification. We observed that while this results in only $4$
benchmarks being in wDNNF in the first pipeline, $219$ benchmarks were
solved in Phase 1 using this pipeline.  This is attributable to
specifications being unate in several output variables, and also
satisfying the condition of Theorem~\ref{lemma:init_sk_good}(a) (while
not being in wDNNF). In the second pipeline, however, we could
represent $230$ benchmarks in wDNNF, and all of these were solved in
Phase 1.

For each benchmark, the order $\preceq$ (ref. step 12 of
Algorithm~\ref{alg:easy}) in which Skolem functions are generated is
such that the variable which occurs in the transitive fan-in of the
least number of nodes in the AIG representation of the specification
is ordered before other variables.  This order ($\preceq$) is used for
both $\bfss$ and $\parsyn$.  Note that the order $\preceq$ is
completely independent of the dynamic variable order used to construct
an ROBDD of the input specification in the second pipeline, prior to
getting the wDNNF representation.

All experiments were performed on a message-passing cluster, with 20
cores and $64$ GB memory per node, each core being a $2.2$ GHz Intel
Xeon processor.  The operating system was Cent OS 6.5.  Twenty cores
were assigned to each run of $\parsyn$.  For $\rsynth$ and $\cadet$ a
single core on the cluster was used, since these tools don't exploit
parallel processing. Each pipeline of $\bfss$ was executed on a single
node; the computation of candidate functions, building of error
formula and refinement of the counterexamples was performed
sequentially on $1$ thread, and $\unigen$ had $19$ threads at its
disposal (idle during Phase 1).

The maximum time given for execution of any run was $3600$ seconds.
The total amount of main memory for any run was restricted to $16$GB.
The metric used to compare the algorithms was {\it time taken to
synthesize Boolean functions}.  The time reported for $\bfss$ is the
better of the two times obtained from the alternative pipelines
described above.  Detailed results from the individual pipelines are
available in Appendix \ref{sec:appendix}.

\noindent{\bf Results.}
Of the $504$ benchmarks, $177$ benchmarks were not solved by any tool
-- $6$ of these being from arithmetic benchmarks and $171$ from
QBFEval.

\begin{table}[ht]
\begin{center}
\begin{scriptsize}
\begin{tabular}{|c|c|c|c|c|c|}
    \hline {Benchmark} & {Total } & {\# Benchmarks}  & {Phase 1} & {Phase 2} & {Solved By} \\ 
     {Domain} &  {Benchmarks} & {Solved} & {Solved} &{Started} & {Phase 2}  \\ 
\hline
{QBFEval} & 383 &   170 & 159 & 73 & 11  \\ 
\hline
{Arithmetic} & 48 & 35  &  35 &  8  & 0   \\ 
\hline
{Disjunctive} &  & & & &   \\
{Decomposition} & 68  & 68 & 66 & 2 & 2   \\ 
\hline
{Factorization} & 5 &  5 & 5 & 0  & 0  \\ 
\hline
  \end{tabular}
\end{scriptsize}
\caption{$\bfss$: Performance summary of combined pipelines}
\label{tab:bfss}
\end{center}
\vspace*{-8mm}
\end{table}
Table \ref{tab:bfss} gives a summary of the performance of $\bfss$
(considering the combined pipelines) over different benchmarks
suites. Of the $504$ benchmarks, $\bfss$ was successful on $278$
benchmarks; of these, $170$ are from QBFEval, $68$ from Disjunctive
Decomposition, $35$ from Arithmetic and $5$ from Factorization.

Of the $383$ benchmarks in the QBFEval suite, we ran $\bfss$ only on
$254$ since we could not build succinct AIGs for the remaining
benchmarks. Of these, \emph{$159$ benchmarks were solved by Phase 1
(i.e., 62\% of built QBFEval benchmarks)} and $73$ proceeded to Phase
2, of which $11$ reached completion. On another $11$ QBFEval
benchmarks Phase 1 timed out.  Of the $48$ Arithmetic
benchmarks, \emph{Phase 1 successfully solved $35$ (i.e., $\sim
72$\%)} and Phase 2 was started for $8$ benchmarks; Phase 1 timed out
on $5$ benchmarks.  Of the $68$ Disjunctive Decomposition
benchmarks, \emph{Phase 1 successfully solved $66$
benchmarks (i.e., 97\%)}, and Phase 2 was started and reached
completion for $2$ benchmarks.  For the $5$ Factorization benchmarks,
Phase 1 was successful on all $5$ benchmarks.

Recall that the goodness ratio is the ratio of the number of {\it
counterexamples remaining} to the {\it total size of the input space}
after Phase 1.  For all benchmarks solved by Phase 1, the goodness
ratio is $0$.  We analyzed the goodness ratio at the beginning of
Phase 2 for $83$ benchmarks for which Phase 2 started.  For $13$
benchmarks this ratio was small $(< 0.002)$, and Phase 2 reached
completion for these.  Of the remaining benchmarks, $34$ also had a
small goodness ratio ($< 0.1$), indicating that we were close to the
solution at the time of timeout.  However, $27$ benchmarks in QBFEval
had goodness ratio  greater than $0.9$, indicating that most of the
counter-examples were not eliminated by timeout.

We next compare the performance of $\bfss$ with other state-of-art
tools. For clarity, since the number of benchmarks in the QBFEval
suite is considerably greater, we plot the QBFEval benchmarks
separately.

\medskip

\begin{figure}[t]
\centering
\begin{subfigure}{2.3in}
  \includegraphics[angle=-90,scale=0.27] {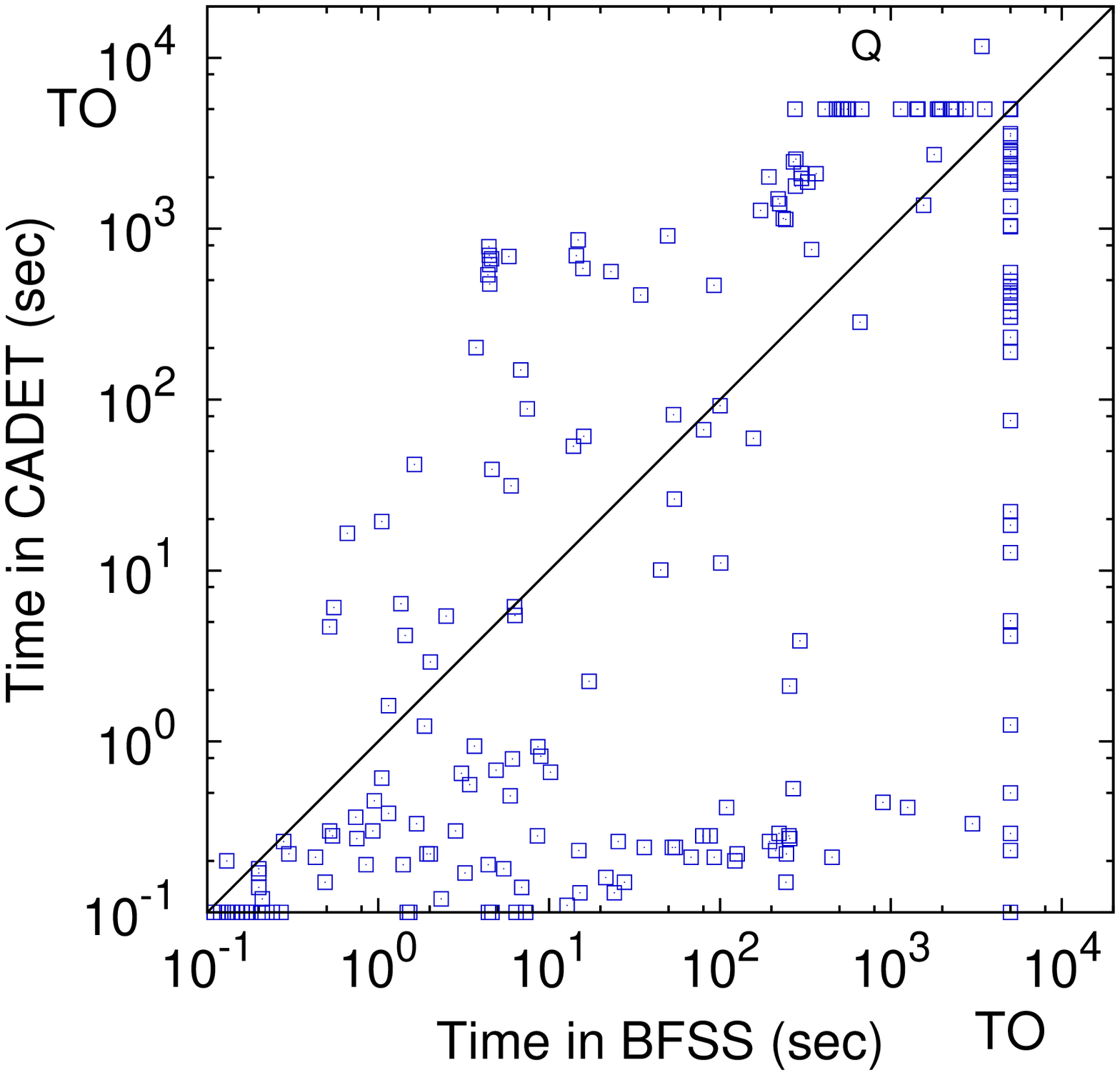} 
\end{subfigure}
\begin{subfigure}{2.3in}
  \includegraphics[angle=-90,scale=0.27] {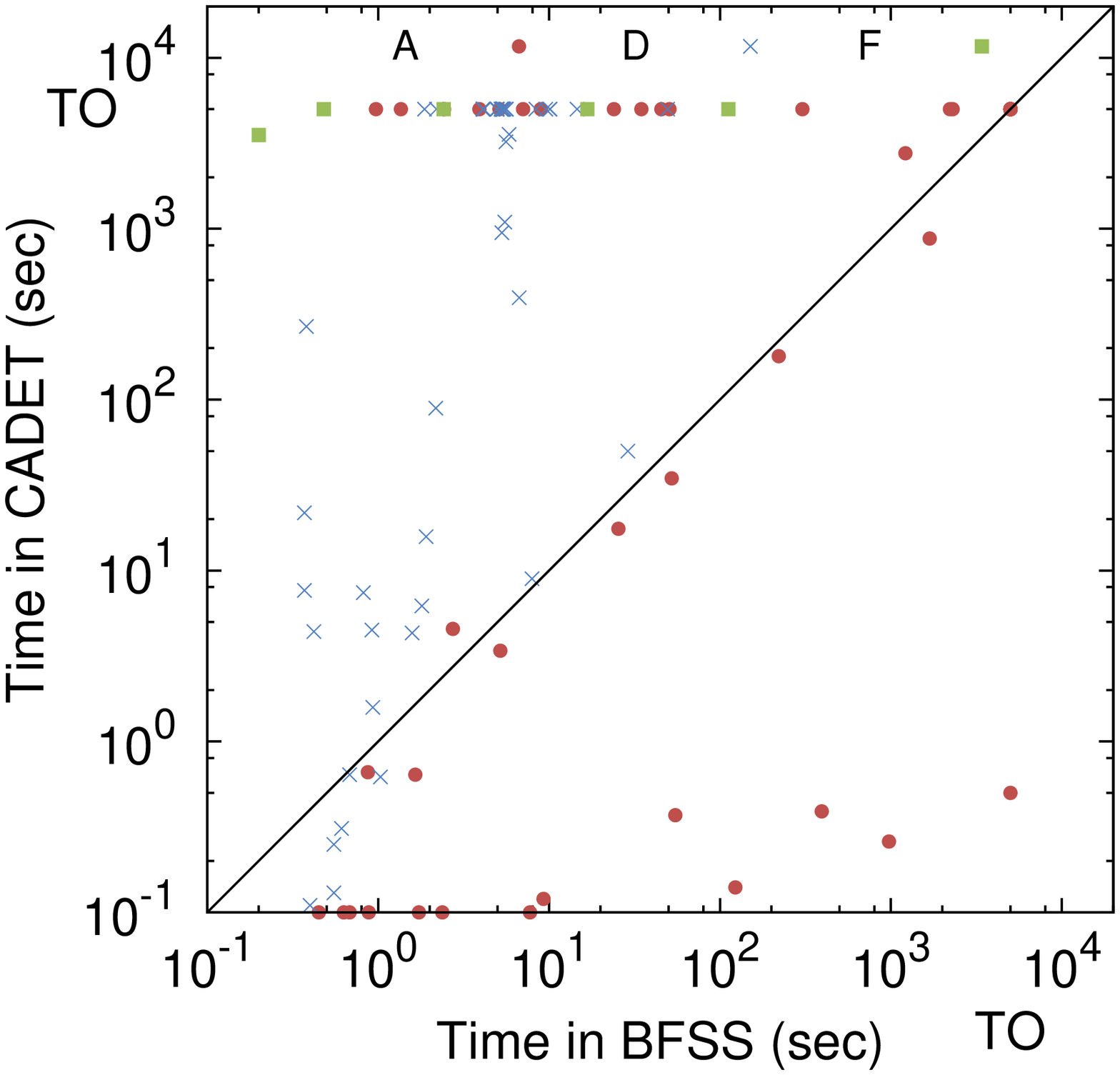} 
\end{subfigure}
\caption{$\bfss$ vs $\cadet$: Legend: \textsf{Q}: QBFEval, \textsf{A}: Arithmetic, \textsf{F}: Factorization, \textsf{D}: Disjunctive Decomposition. \textsf{TO}: benchmarks for which the corresponding algorithm was unsuccessful.}
\label{fig:bfsscadet}
\end{figure}

\noindent { $\bfss$ vs $\cadet$}: 
Of the $504$ benchmarks, $\cadet$ was successful on $231$ benchmarks,
of which $24$ belonged to Disjunctive Decomposition, $22$ to
Arithmetic, $1$ to Factorization and $184$ to QBFEval.
Figure \ref{fig:bfsscadet}(a) gives the performance of the two
algorithms with respect to time on the QBFEval suite. Here, $\cadet$
solved $35$ benchmarks that $\bfss$ could not solve, whereas $\bfss$
solved $21$ benchmarks that could not be solved by $\cadet$.
Figure \ref{fig:bfsscadet}(b) gives the performance of the two
algorithms with respect to time on the Arithmetic, Factorization and
Disjunctive Decomposition benchmarks.  In these categories, there were
a total of $62$ benchmarks that $\bfss$ solved that $\cadet$ could not
solve, and there was $1$ benchmark that $\cadet$ solved but $\bfss$
did not solve.  While $\cadet$ takes less time on Arithmetic
benchmarks and many QBFEval benchmarks, on Disjunctive Decomposition
and Factorization, $\bfss$ takes less time.
\begin{figure}[t]
\centering
\begin{subfigure}{2.3in}
  \includegraphics[angle=-90,scale=0.27] {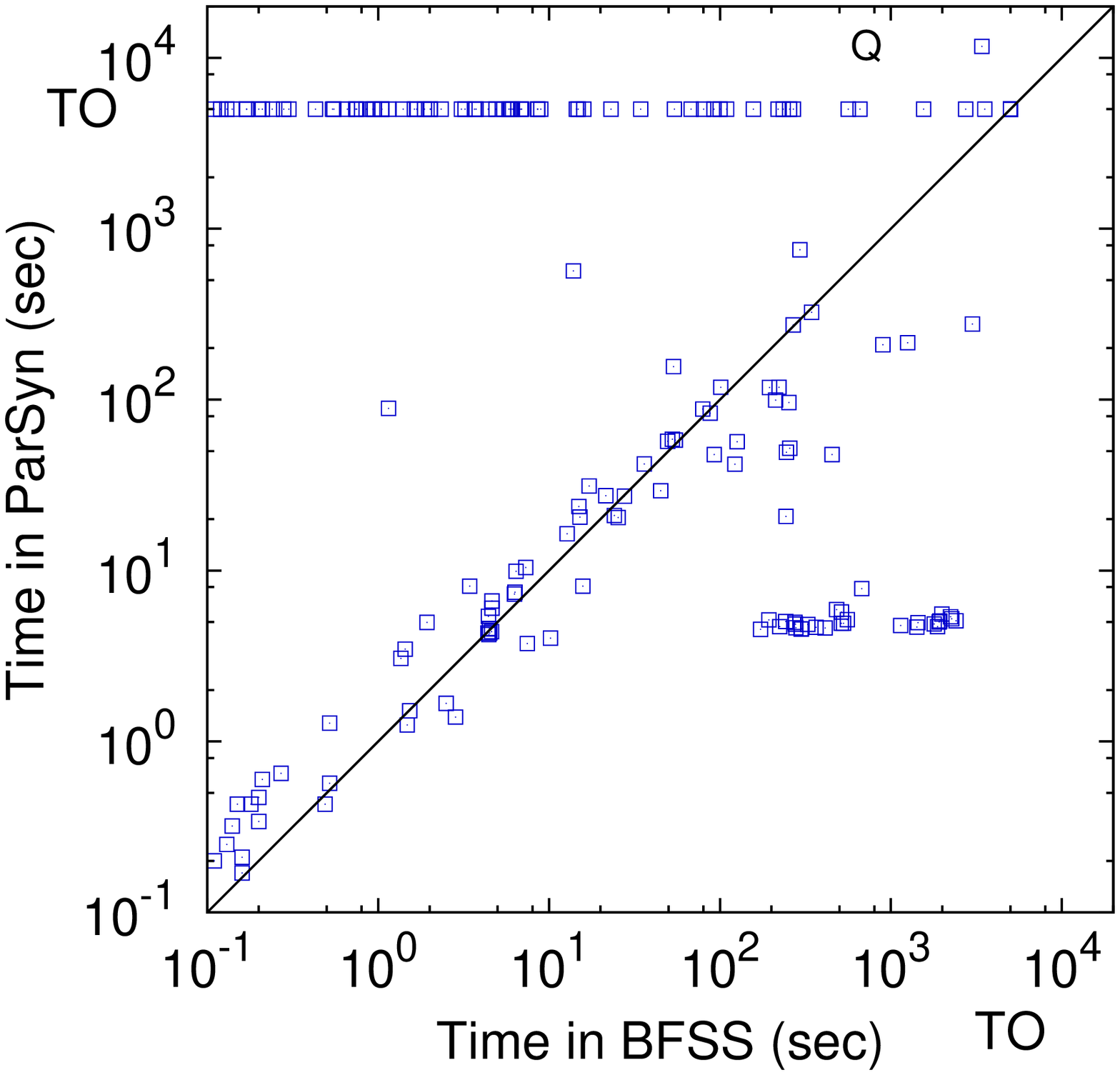} 
\end{subfigure}
\begin{subfigure}{2.3in}
  \includegraphics[angle=-90,scale=0.27] {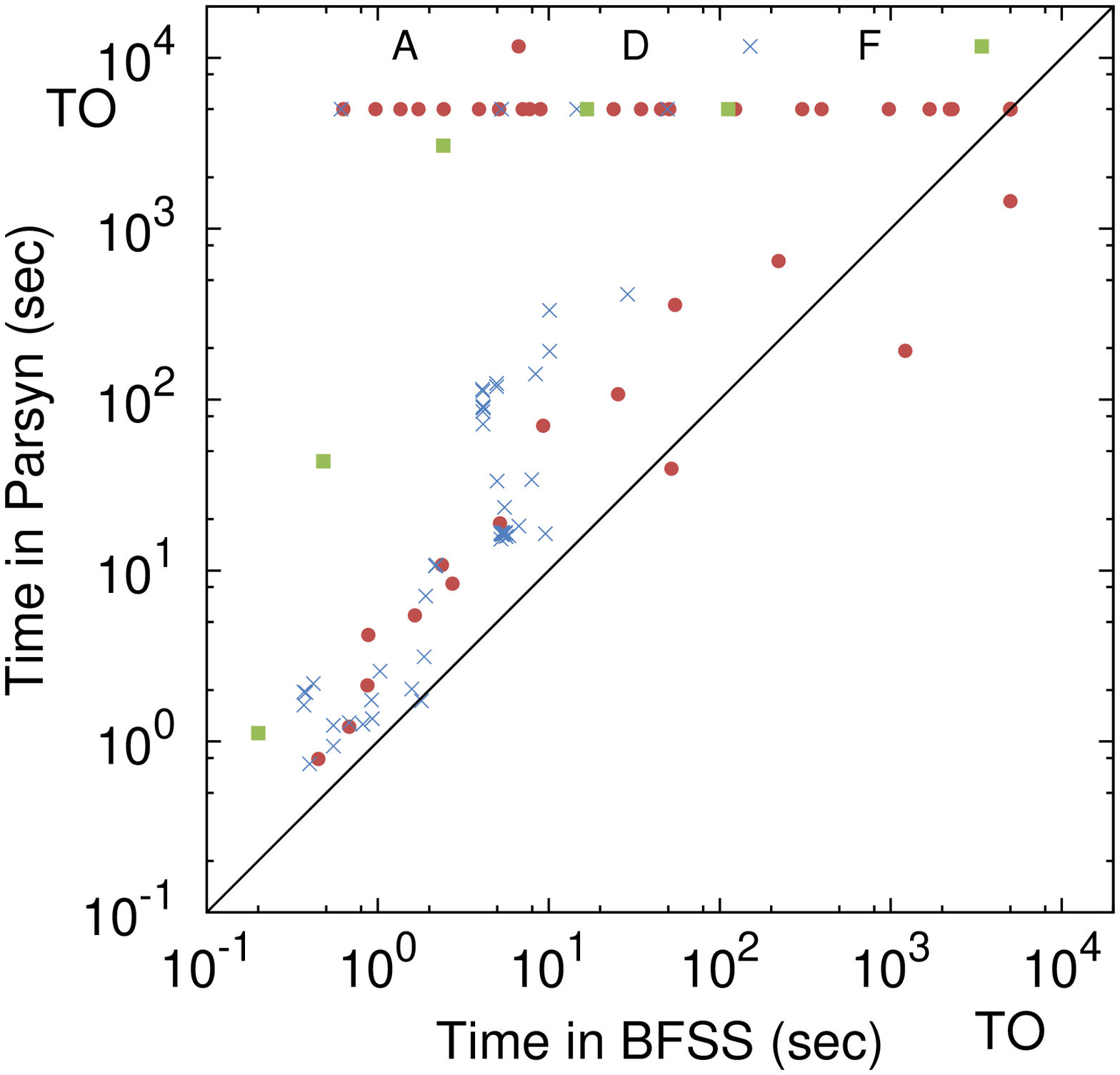} 
\label{fig:bfssparsyntacmis}
\end{subfigure}
\caption{$\bfss$ vs $\parsyn$ (for legend see Figure \ref{fig:bfsscadet})}
\label{fig:bfssparsyn}
\end{figure}

\noindent \bfss\  vs \parsyn:
Figure \ref{fig:bfssparsyn} shows the comparison of time taken by
$\bfss$ and $\parsyn$.  $\parsyn$ was successful on a total of $185$
benchmarks, and could solve $1$ benchmark which $\bfss$ could not
solve.  On the other hand, $\bfss$ solved $94$ benchmarks that $\parsyn$ could not
solve.  From Figure \ref{fig:bfssparsyn}, we can see that on most of
the Arithmetic, Disjunctive Decomposition and Factorization
benchmarks, $\bfss$ takes less time than $\parsyn$.

\begin{figure}[t]
\centering
\begin{subfigure}{2.3in}
  \includegraphics[angle=-90,scale=0.27] {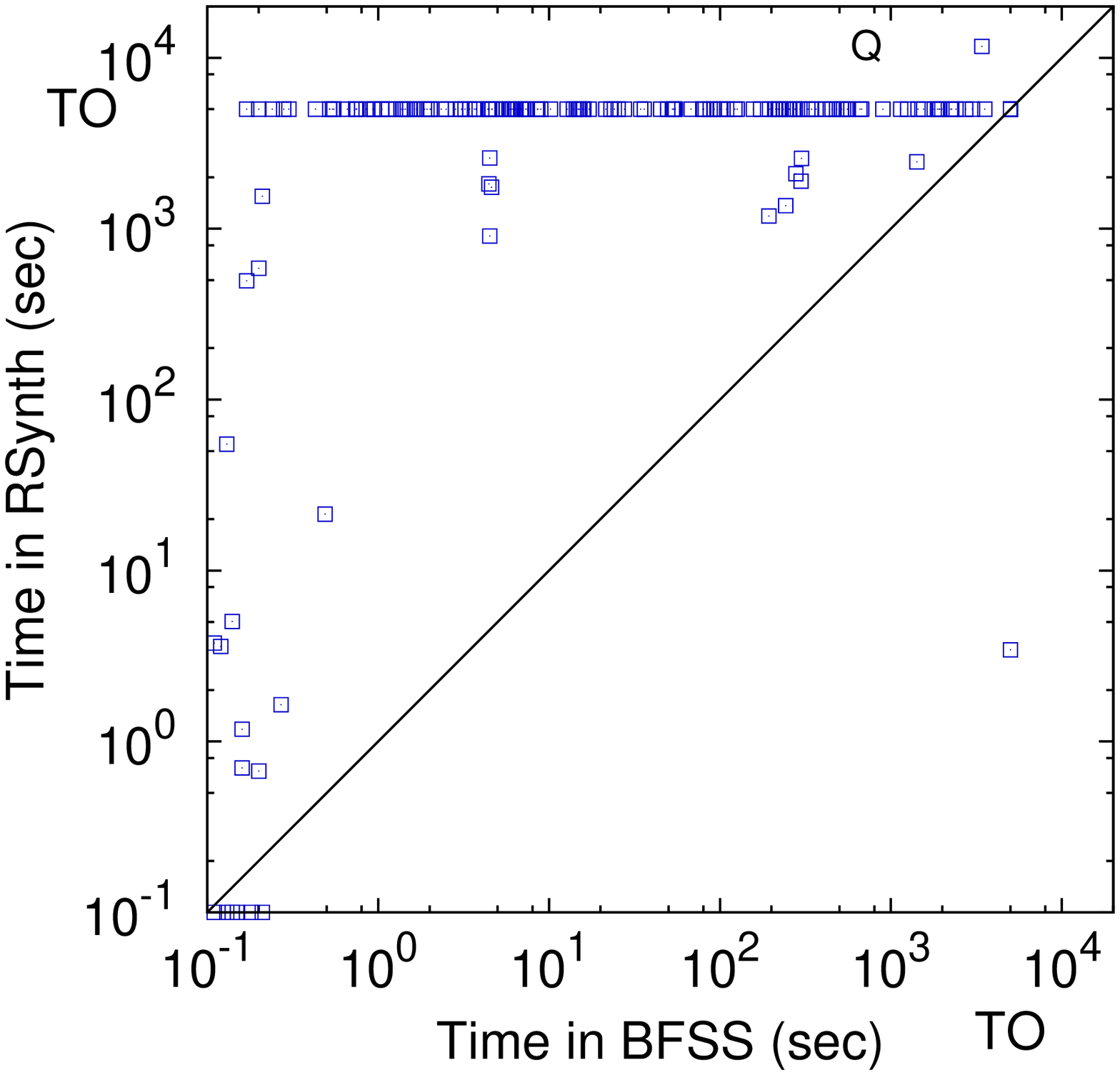} 
\end{subfigure}
\begin{subfigure}{2.3in}
  \includegraphics[angle=-90,scale=0.27] {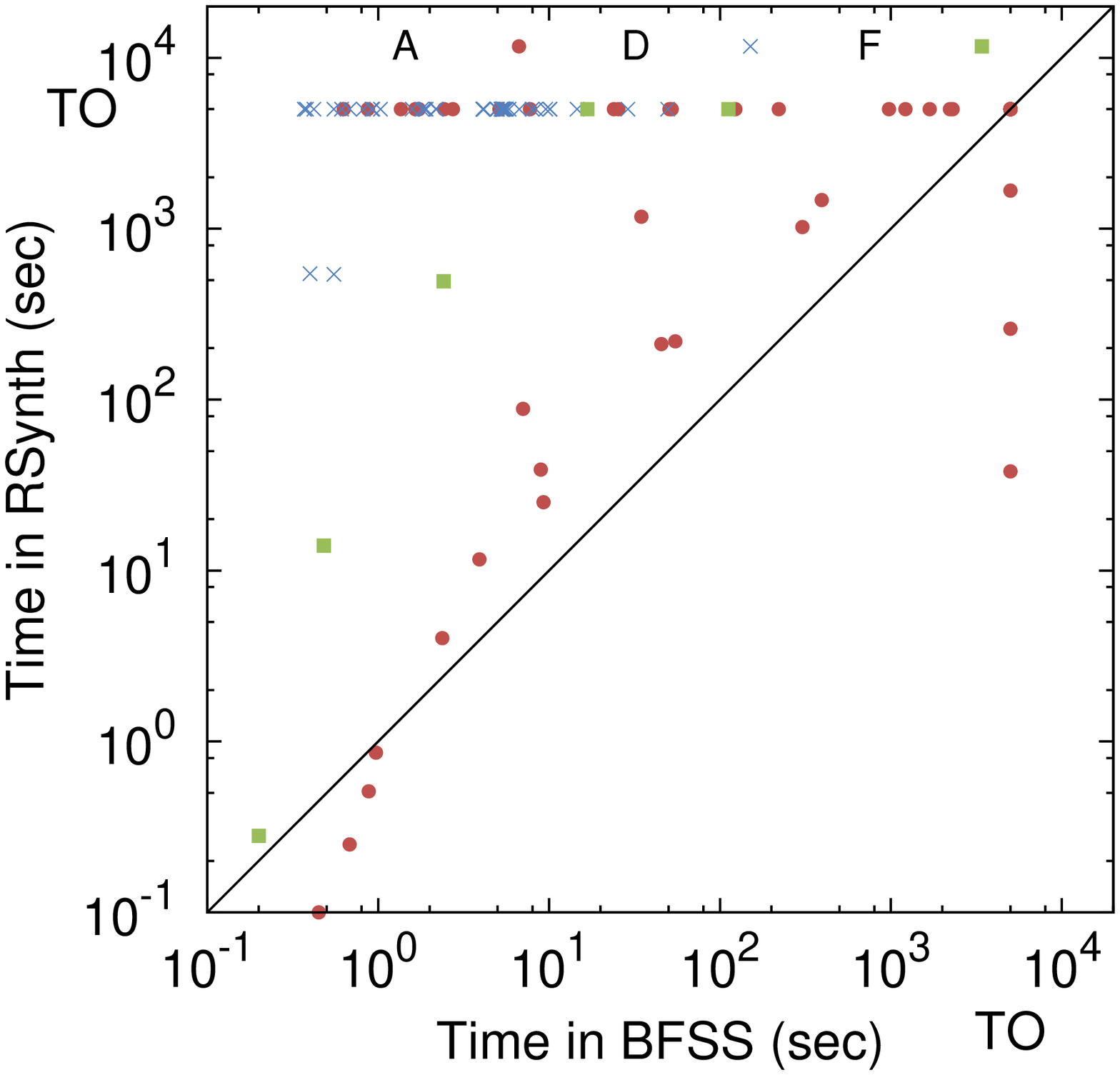} 
\end{subfigure}
\caption{$\bfss$ vs $\rsynth$ (for legend see Figure \ref{fig:bfsscadet})} 
\label{fig:bfssrsynth}
\end{figure}

\noindent { $\bfss$ vs $\rsynth$}: We next compare the performance of $\bfss$ with $\rsynth$. As shown in Figure \ref{fig:bfssrsynth}, $\rsynth$ was successful on $51$ benchmarks,
with $4$ benchmarks that could be solved by $\rsynth$ but not by
$\bfss$. In contrast, $\bfss$ could solve $231$ benchmarks that
$\rsynth$ could not solve!  Of the benchmarks that were solved by both
solvers, we can see that $\bfss$ took less time on most of them.

\begin{figure}[b]
\centering
\begin{subfigure}{2.3in}
  \includegraphics[angle=-90,scale=0.26] {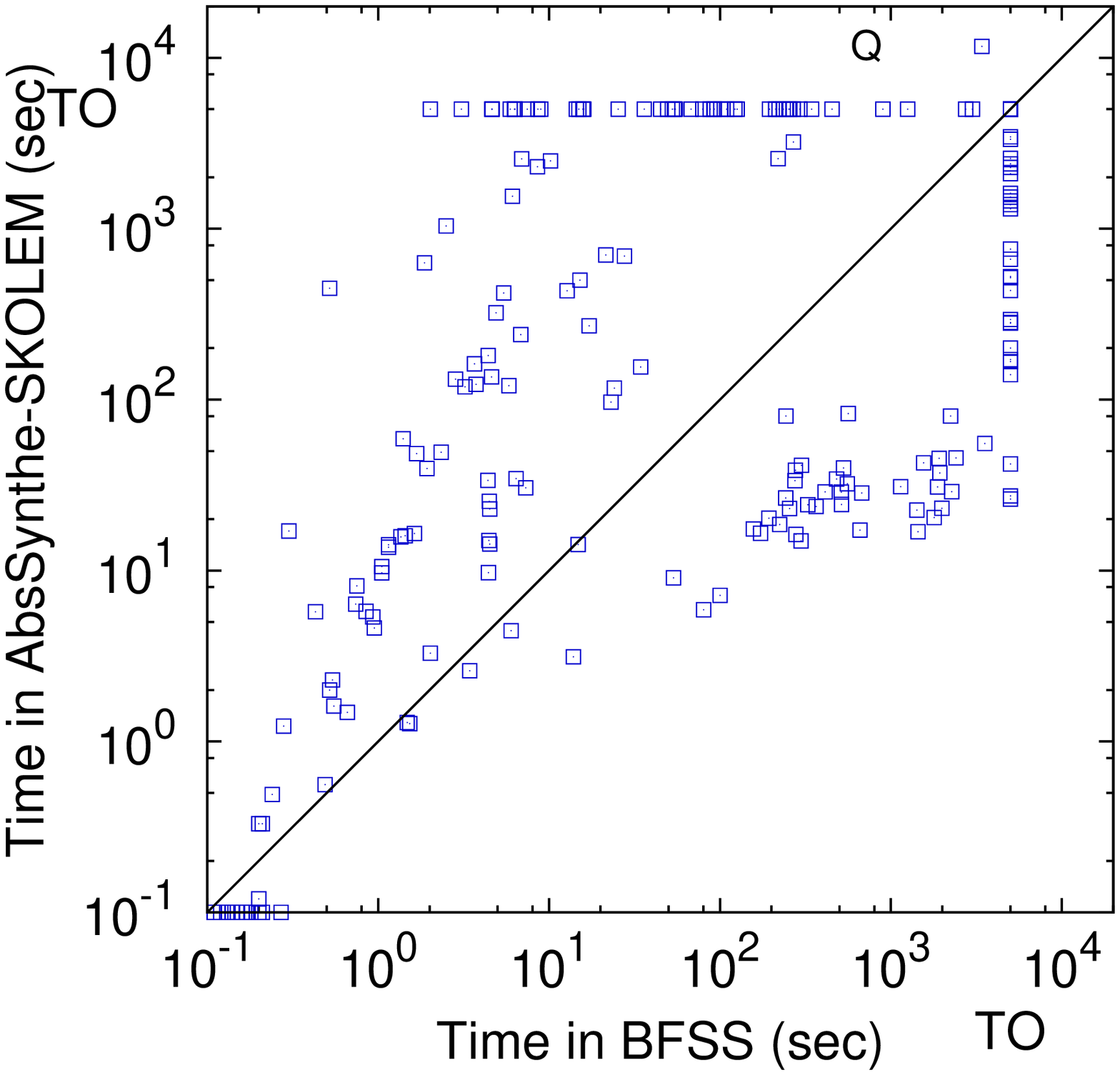} 
\end{subfigure}
\begin{subfigure}{2.3in}
  \includegraphics[angle=-90,scale=0.26] {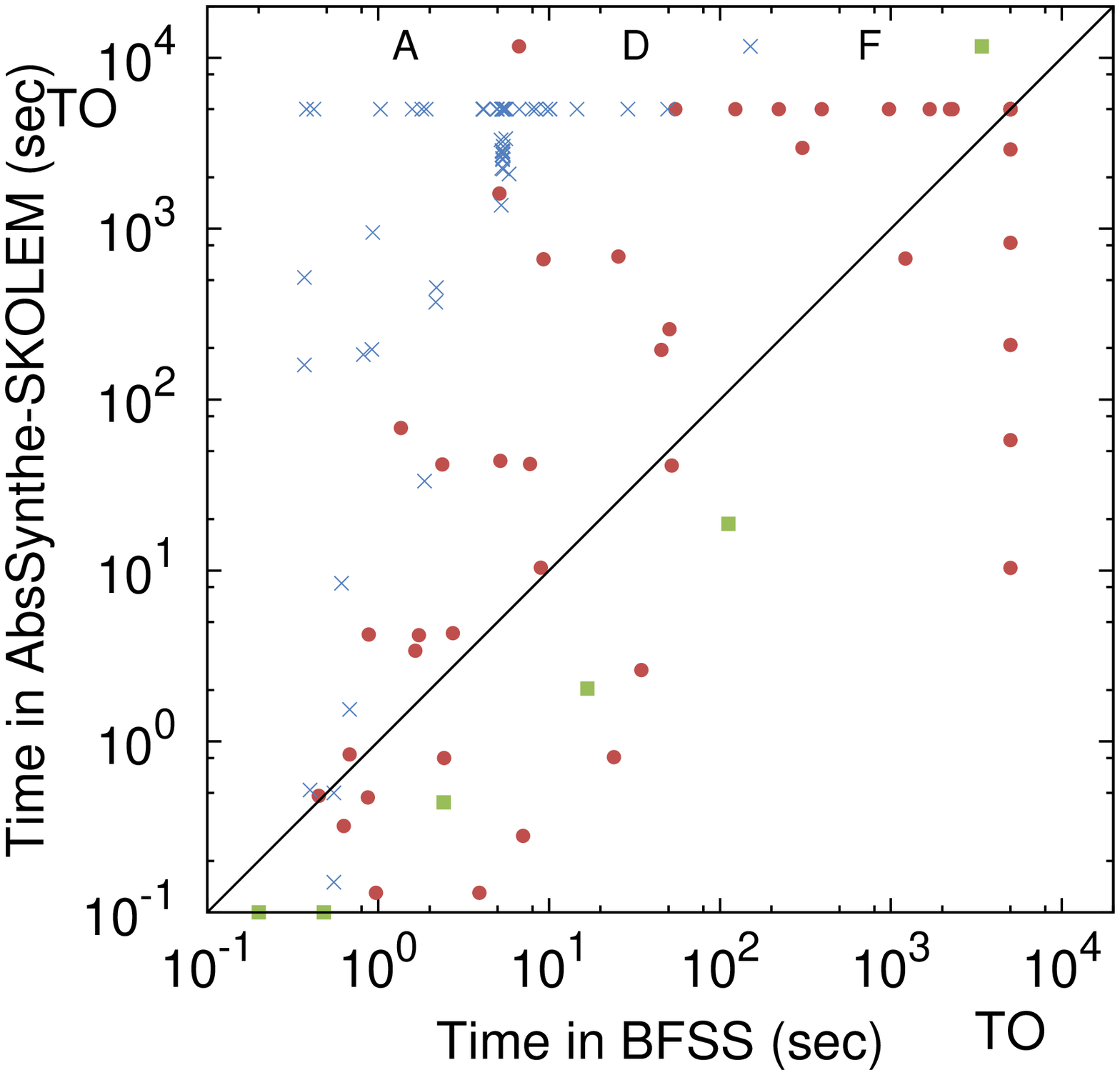} 
\end{subfigure}
\caption{$\bfss$ vs $\abssynsk$ (for legend see Figure \ref{fig:bfsscadet})} 
\label{fig:bfssabs}
\end{figure}
\noindent { $\bfss$ vs $\abssynsk$}: $\abssynsk$ was successful on $217$ benchmarks, and
          could solve $31$ benchmarks that $\bfss$ could not solve. In contrast,
$\bfss$ solved a total of $92$ benchmarks that $\abssynsk$ could not.
Figure~\ref{fig:bfssabs} shows a comparison of running times of
$\bfss$ and $\abssynsk$.

\section{Conclusion}
\label{sec:concl}
\vspace*{-3mm}
In this paper, we showed some complexity-theoretic hardness results
for the Boolean functional synthesis problem.  We then developed a
two-phase approach to solve this problem, where the first phase, which
is an efficient algorithm generating poly-sized functions surprisingly
succeeds in solving a large number of benchmarks. To explain this, we
identified sufficient conditions when phase 1 gives the correct
answer. For the remaining benchmarks, we employed the second phase of
the algorithm that uses a CEGAR-based approach and builds Skolem
functions by exploiting recent advances in SAT solvers/approximate
counters.  As future work, we wish to explore further improvements in
Phase 2, and other structural restrictions on the input
that ensure completeness of Phase 1.

\paragraph{\bfseries Acknowledgements:} We are thankful to Ajith John, Kuldeep Meel, Mate Soos,
Ocan Sankur, Lucas Martinelli Tabajara and Markus Rabe for useful discussions and for providing us with
various software tools used in the experimental comparisons. We also thank the anonymous reviewers
for insightful comments.

\bibliographystyle{splncs03}
\bibliography{ref}
\newpage
\appendix
\section{Detailed Results for individual pipelines of BFSS}
\label{sec:appendix}

As mentioned in section \ref{sec:expt}, $\bfss$ is an ensemble of two pipelines, an AIG-NNF pipeline and a BDD-wDNNF pipeline.
These two pipelines accept the
same input specification but represent them in two different ways.
The first pipeline takes the input formula as an AIG and builds an NNF
(not necessarily a wDNNF) DAG, while the second pipeline first builds an ROBDD
from the input AIG using dynamic variable reordering, and then obtains a wDNNF representation from the ROBDD
using the linear-time algorithm described in~\cite{darwiche-jacm}.
Once the NNF/wDNNF representation is built, the same algorithm is used to generate skolem functions, namely,
  Algorithm ~\ref{alg:easy} is used in Phase 1 and CEGAR-based synthesis using
$\unigen$\cite{unigen2} to sample counterexamples is used in Phase 2. In this section, we give the individual results of the two pipelines.

\subsection{Performance of the AIG-NNF pipeline}

\begin{table}[ht]
\begin{center}
\begin{tabular}{|c|c|c|c|c|c|}
    \hline {Benchmark} & {Total } & {\# Benchmarks}  & {Phase 1} & {Phase 2} & {Solved By} \\ 
     {Domain} &  {Benchmarks} & {Solved} & {Solved} &{Started} & {Phase 2}  \\ 
\hline
{QBFEval} & 383 & 133 & 122 & 110 & 11 \\ 
\hline
{Arithmetic} & 48 & 31  &  31 & 12 & 0  \\ 
\hline
{Disjunctive} &  & & & &   \\
{Decomposition} & 68  & 68 & 66 & 2 & 2  \\ 
\hline
{Factorization} & 5 & 4 & 0 & 5 & 4  \\ 
\hline
  \end{tabular}
\caption{$\bfss$: Performance Summary for AIG-NNF pipeline }
\label{tab:bfss2}
\end{center}
\end{table}

In the AIG-NNF pipeline, $\bfss$ solves a total of $236$ benchmarks, with $133$ benchmarks in QBFEval, $31$ in Arithmetic, all the $68$ benchmarks of Disjunctive Decomposition and $4$ benchmarks in Factorization.
Of the $254$ benchmarks in QBFEval (as mentioned in Section \ref{sec:expt}, we could not build succinct AIGs for the remaining benchmarks and did not run our tool on them), Phase 1 solved $122$ benchmarks and Phase 2 was started on $110$ benchmarks, of which $11$ benchmarks reached completion. Of the $48$ benchmarks in Arithmetic, Phase 1 solved $31$ and Phase $2$ was started on $12$. On the remaining $5$ Arithmetic benchmarks, Phase 1 did not reach completion. Of the $68$ Disjunctive Decomposition benchmarks, $66$ were successfully solved by Phase 1 and the remaining $2$ by Phase 2. Phase 2 had started on all the $5$ benchmarks in Factorization and  reached completion on $4$ benchmarks.

\subsubsection{Plots for the AIG-NNF pipeline}

\begin{figure}[h]
\begin{subfigure}{2.3in}
  \hspace{-1cm}
  \includegraphics[angle=-90,scale=0.28] {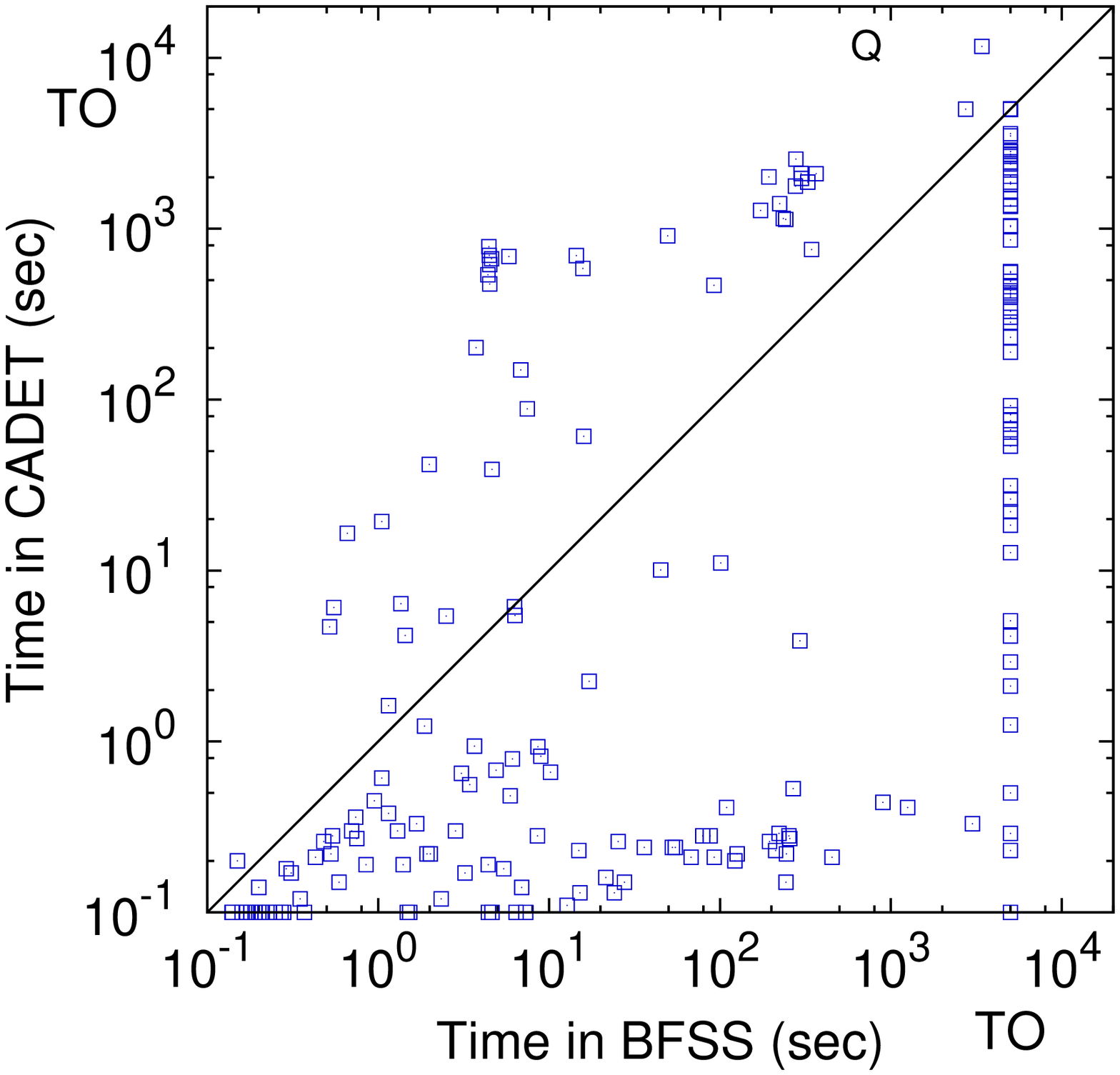} 
\end{subfigure}
\begin{subfigure}{2.3in}
  \includegraphics[angle=-90,scale=0.28] {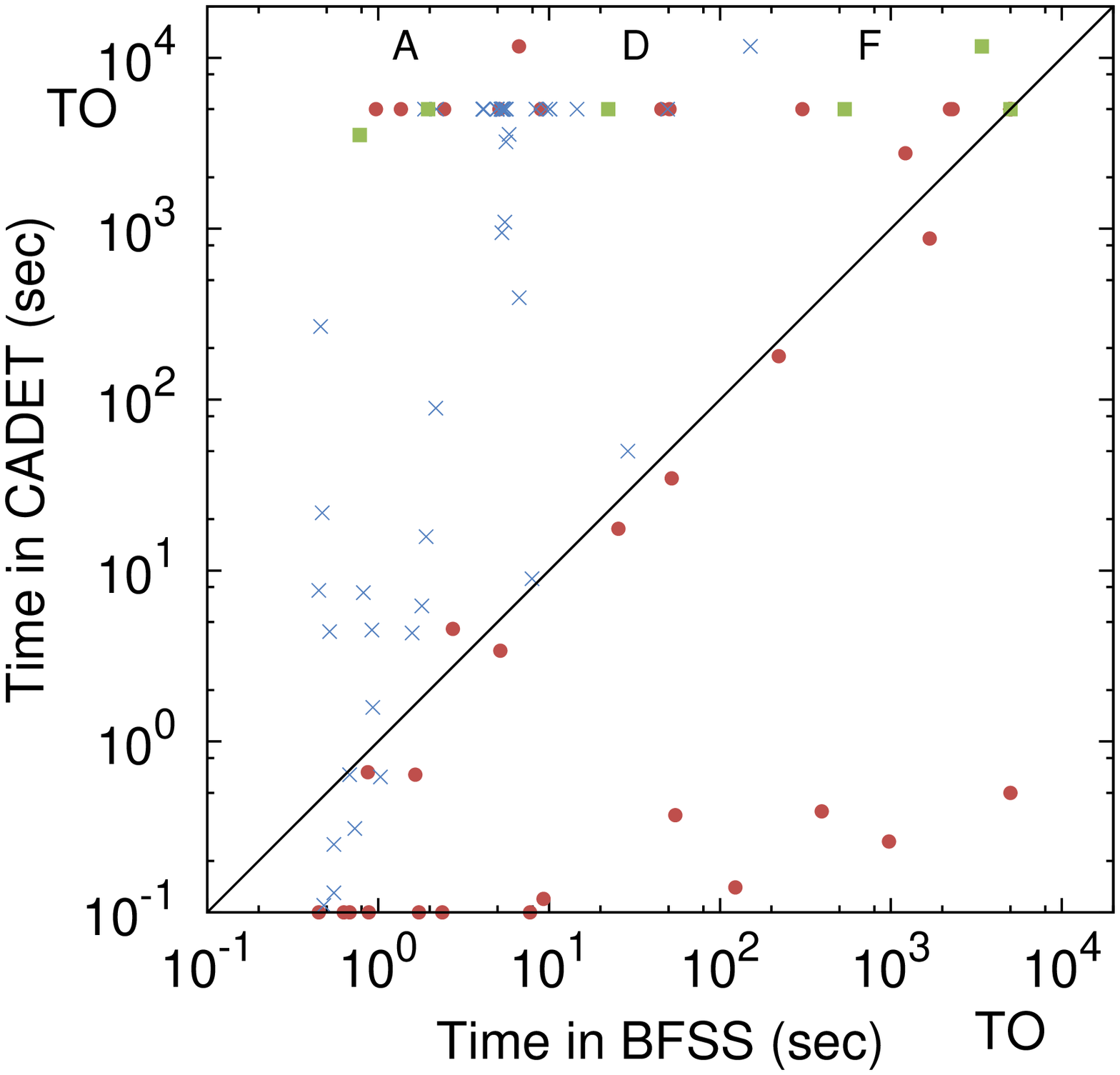} 
\end{subfigure}
\caption{$\bfss$ (AIG-NNF Pipeline) vs $\cadet$. Legend : \texttt{A}: Arithmetic, \texttt{F}: Factorization, \texttt{D}: Disjunctive decomposition \texttt{Q} QBFEval. \texttt{TO}: benchmarks for which the corresponding algorithm was unsuccessful.}
\label{fig:bfsscadetAIG}
\end{figure}

Figure \ref{fig:bfsscadetAIG} shows the performance of $\bfss$ (AIG-NNF pipeline) versus $\cadet$ for all the four benchmark domains. Amongst the four domains, $\cadet$ solved $53$ benchmarks that $\bfss$ could not solve. Of these, $52$ belonged to QBFEval and $1$ belonged to Arithmetic. On the other hand, $\bfss$ solved $58$ benchmarks that $\cadet$ could not solve. Of these, $1$ belonged to QBFEval, $10$ to Arithmetic, $3$ to Factorization and $44$ to Disjunctive Decomposition. From Figure \ref{fig:bfsscadetAIG}, we can see that while $\cadet$ takes less time than $\bfss$ on many Arithmetic and QBFEval benchmarks, on Disjunctive Decomposition and Factorization, the AIG-NNF pipeline of $\bfss$ takes less time.

\begin{figure}[h]
\centering
\begin{subfigure}{2.3in}
  \hspace{-1cm}
  \includegraphics[angle=-90,scale=0.28] {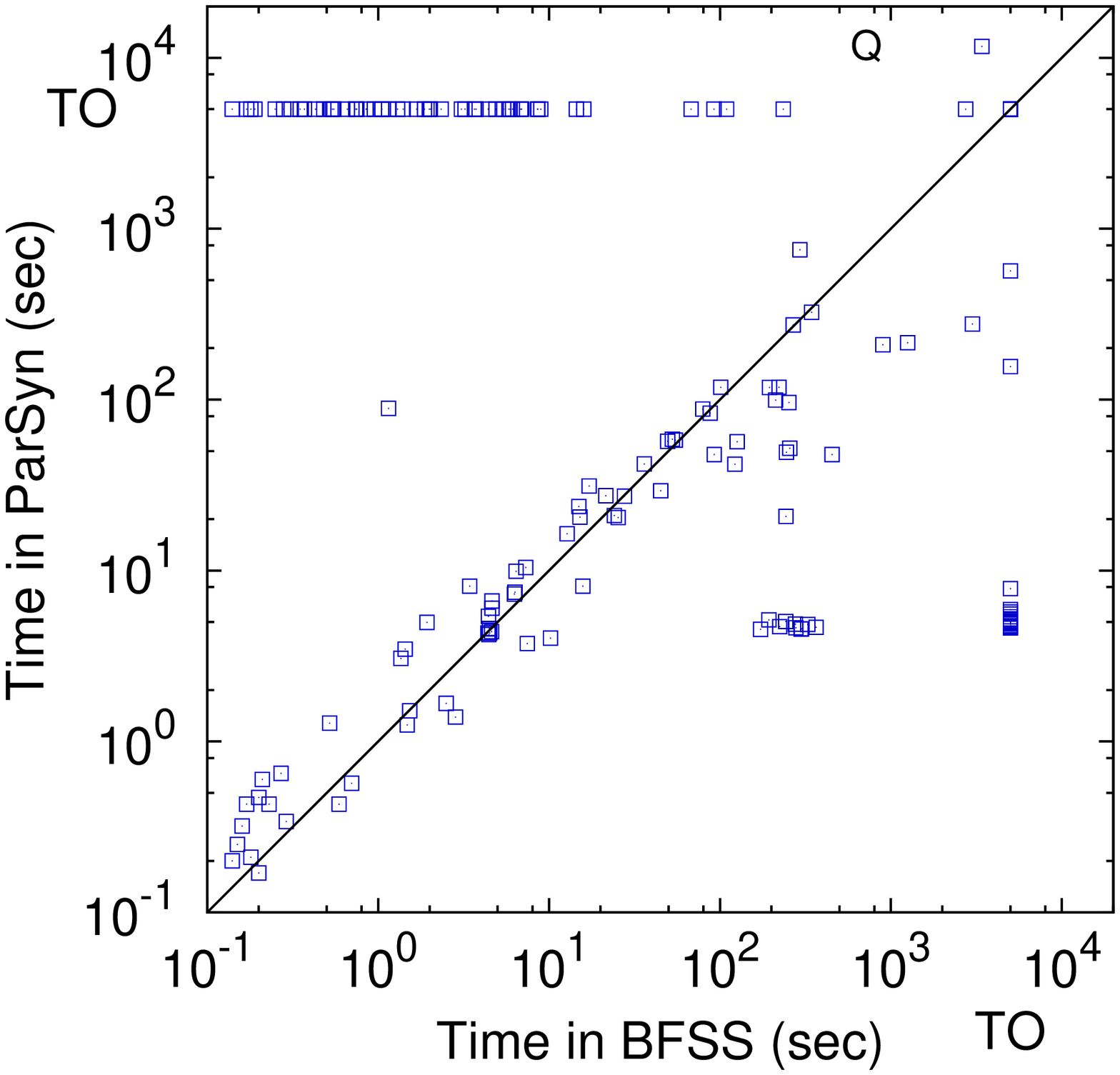} 
\end{subfigure}
\begin{subfigure}{2.3in}
  \includegraphics[angle=-90,scale=0.28] {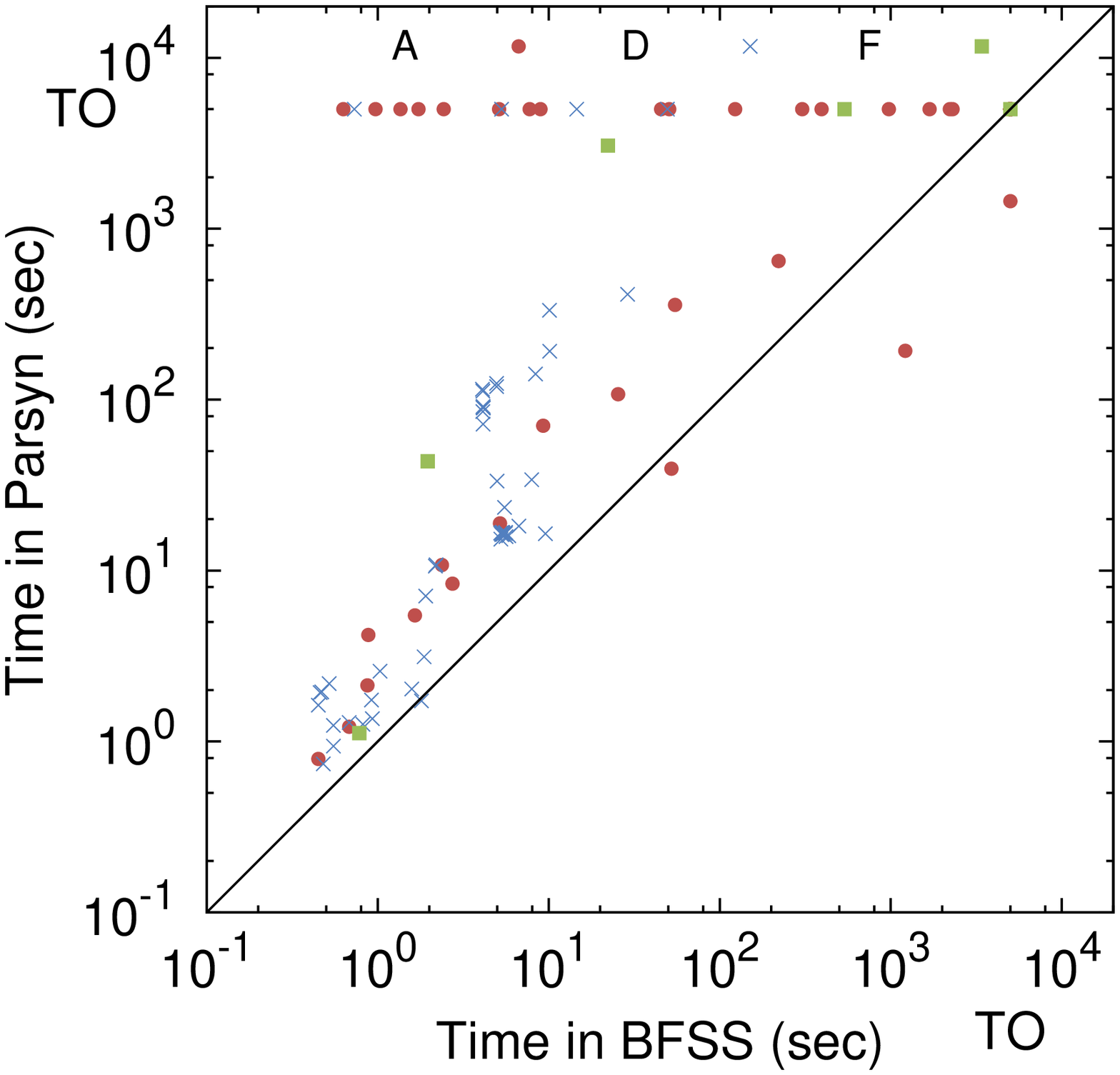} 
\end{subfigure}
\caption{$\bfss$ (AIG-NNF Pipeline) vs $\parsyn$ (for legend see Figure \ref{fig:bfsscadetAIG})}
\label{fig:bfssparsynAIG}
\end{figure}

Figure \ref{fig:bfssparsynAIG} shows the performance of $\bfss$ (AIG-NNF pipeline) versus $\parsyn$. Amongst the $4$ domains, $\parsyn$ solved $22$ benchmarks that $\bfss$ could not solve, of these $1$ benchmark belonged to the Arithmetic domain and $21$ benchmarks belonged to QBFEval. On the other hand, $\bfss$ solved $73$ benchmarks that $\parsyn$ could not solve. Of these, $51$ belonged to QBFEval, $17$ to Arithmetic and $4$ to Disjunctive Decomposition. From \ref{fig:bfsscadetAIG}, we can see that while the behaviour of $\parsyn$ and $\bfss$ is comparable for many QBFEval benchmarks, on most of the Arithmetic, Disjunctive Decomposition and Factorization benchmarks, the AIG-NNF pipeline of $\bfss$ takes less time.
\begin{figure}[h]
\centering
\begin{subfigure}{2.3in}
  \hspace{-1cm}
  \includegraphics[angle=-90,scale=0.28] {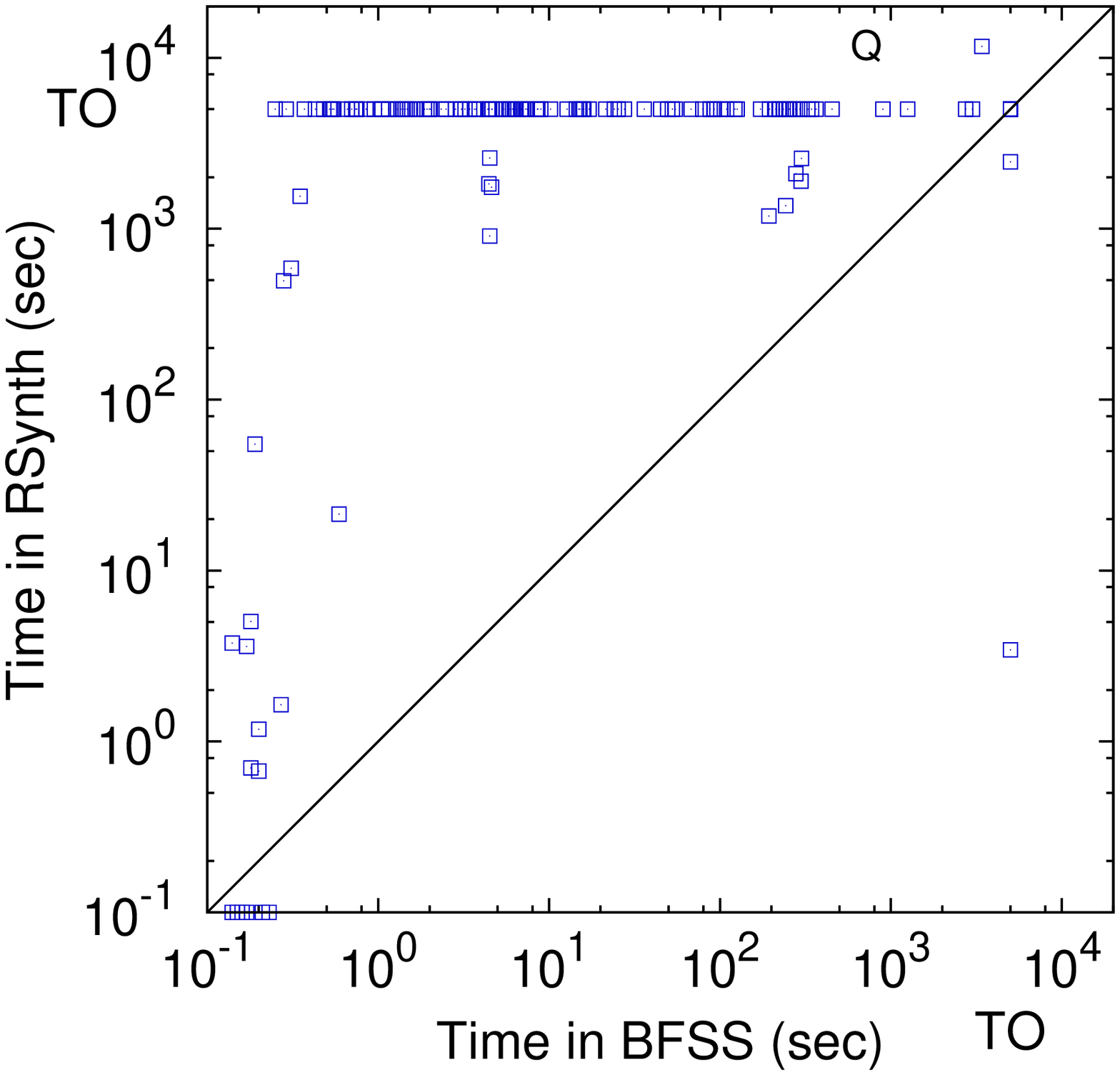} 
\end{subfigure}
\begin{subfigure}{2.3in}
  \includegraphics[angle=-90,scale=0.28] {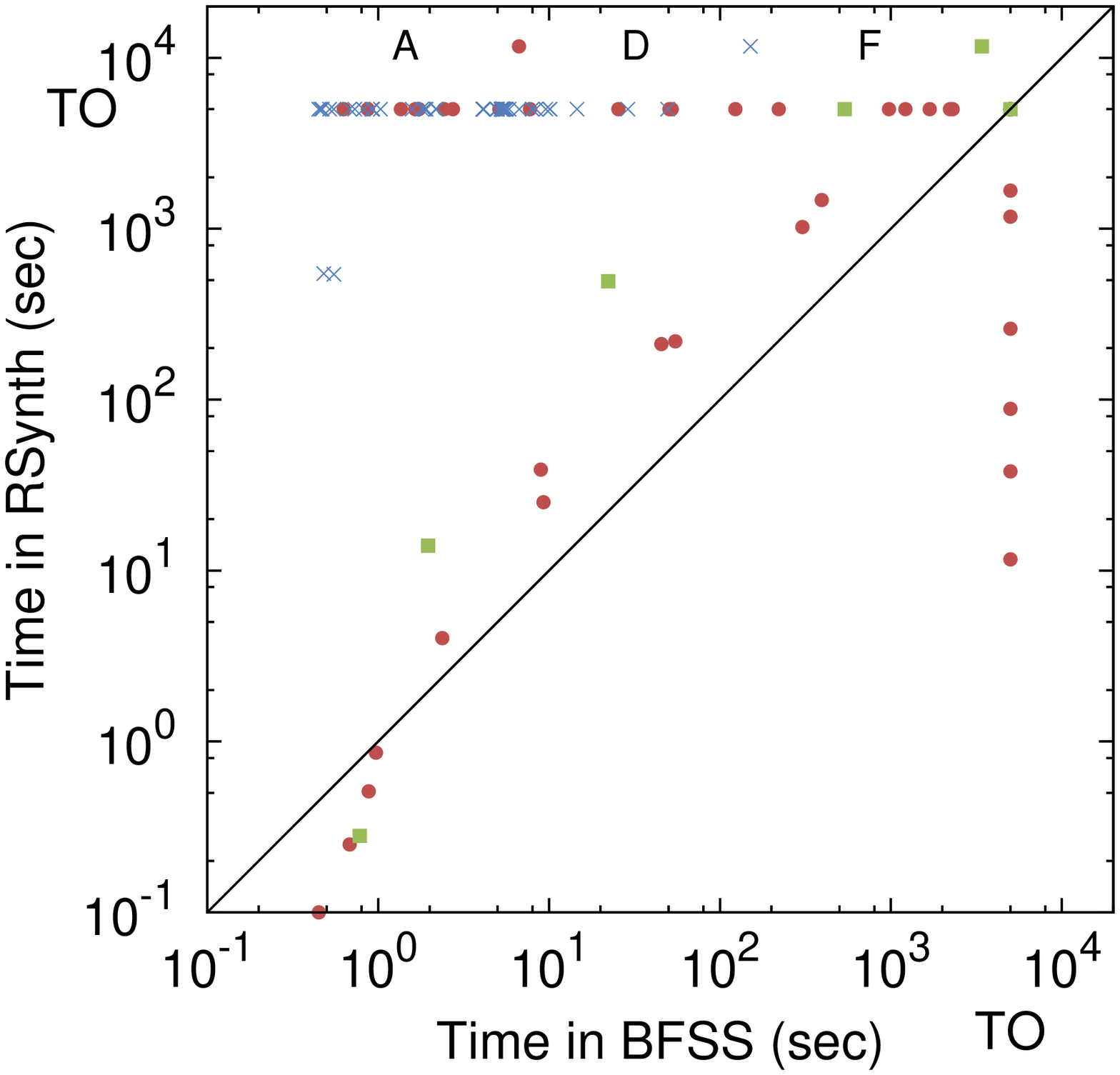} 
\end{subfigure}
\caption{$\bfss$ (AIG-NNF Pipeline) vs $\rsynth$ (for legend see Figure \ref{fig:bfsscadetAIG})} 
\label{fig:bfssrsynthAIG}
\end{figure}

Figure \ref{fig:bfssrsynthAIG} gives the comparison of the AIG-NNF pipeline of $\bfss$ and $\rsynth$. While $\rsynth$ solves $8$ benchmarks that $\bfss$ does not solve, 
$\bfss$ solves $193$ benchmarks that $\rsynth$ could not solve. Of these $106$ belonged to QBFEval, $20$ to Arithmetic, $66$ to Disjunctive Decomposition and $1$ to Factorization. Moreover, on most of the benchmarks that both the tools solved, $\bfss$ takes less time.

\begin{figure}[h]
\centering
\begin{subfigure}{2.3in}
  \hspace{-1cm}
  \includegraphics[angle=-90,scale=0.28] {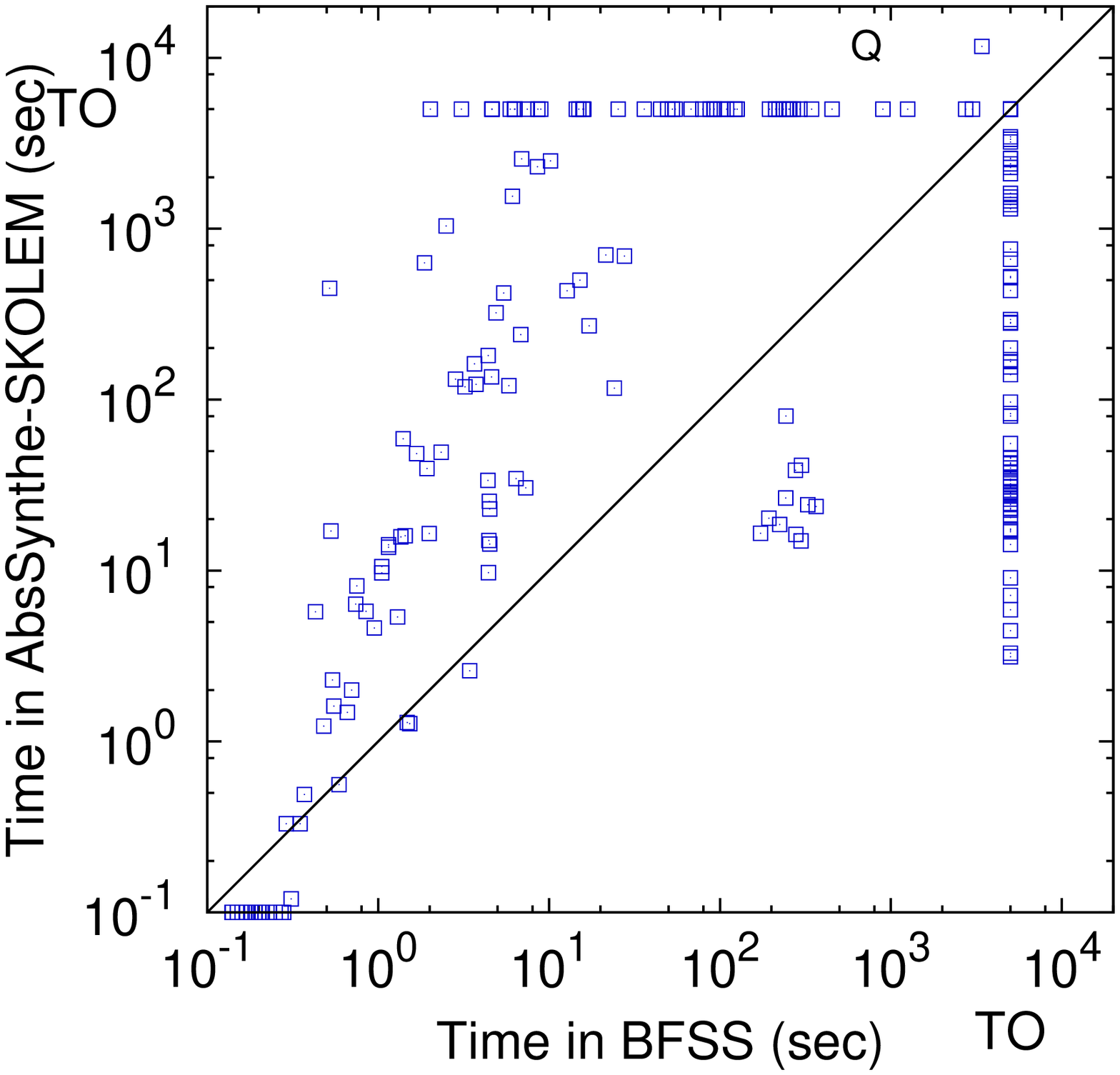} 
\end{subfigure}
\begin{subfigure}{2.3in}
  \includegraphics[angle=-90,scale=0.28] {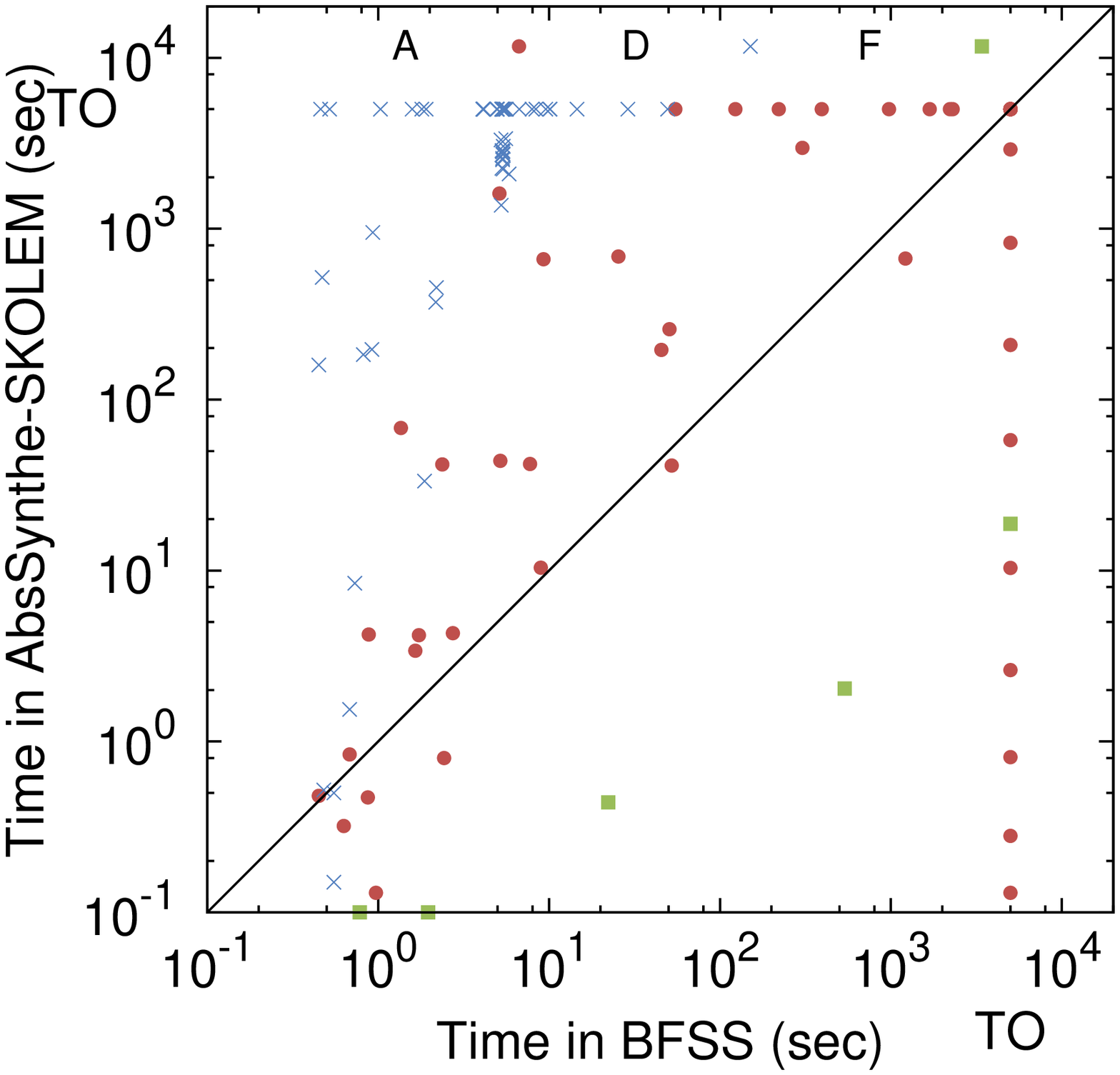} 
\end{subfigure}
\caption{$\bfss$ (AIG-NNF Pipeline) vs $\abssynsk$ (for legend see Figure \ref{fig:bfsscadetAIG})} 
\label{fig:bfssabsAIG}
\end{figure}

Figure \ref{fig:bfssabsAIG} gives the comparison of of the performance of the AIG-NNF pipeline of $\bfss$ and $\abssynsk$. While $\abssynsk$ solves $72$ benchmarks that $\bfss$ could not solve, $\bfss$ solved $91$ benchmarks that $\abssynsk$ could not solve. Of these $44$ belonged to QBFEval, $8$ to Arithmetic and $39$ to Disjunctive Decomposition.

\subsection{Performance of the BDD-wDNNF  pipeline}

In this section, we discuss the performance of the BDD-wDNNF pipeline of $\bfss$.
Recall that in this pipeline the tool builds an ROBDD
from the input AIG using dynamic variable reordering and then converts the ROBDD in a wDNNF representation. 
In this section, by $\bfss$ we mean, the BDD-wDNNF pipeline of the tool.

Table \ref{tab:bfss3} gives the performance summary of the BDD-wDNNF pipeline. Using this pipeline, the tool solved a total of $230$ benchmarks, of which $143$ belonged to QBFEval, $23$ belonged to Arithmetic, $59$ belonged to Disjunctive Decomposition and $5$ belonged to Factorization. As expected, since the representation is already in wDNNF, the skolem functions generated at end of Phase 1 were indeed exact (see Theorem \ref{lemma:init_sk_good}(b)) and we did not require to start Phase 2 on any benchmark. We also found that the memory requirements of this pipeline were higher, and for some benchmarks the tool failed because the ROBDDs (and hence resulting wDNNF representation) were large in size, resulting in out of memory errors or assertion failures in the underlying AIG library. 


\begin{table}[ht]
\begin{center}
\begin{tabular}{|c|c|c|c|c|c|}
    \hline {Benchmark} & {Total } & {\# Benchmarks}  & {Phase 1} & {Phase 2} & {Solved By} \\ 
     {Domain} &  {Benchmarks} & {Solved} & {Solved} &{Started} & {Phase 2}  \\ 
\hline
{QBFEval} & 383 & 143  & 143 & 0 & 0   \\ 
\hline
{Arithmetic} & 48 & 23  &  23 &  0 & 0   \\ 
\hline
{Disjunctive} &  & & & &   \\
{Decomposition} & 68  & 59 & 59 & 0 & 0   \\ 
\hline
{Factorization} & 5 &  5 & 5 & 0  & 0  \\ 
\hline
  \end{tabular}
\caption{$\bfss$: Performance Summary }
\label{tab:bfss3}
\end{center}
\end{table}

\subsubsection{Plots for the BDD-wDNNF pipeline} 

\begin{figure}[h]
\centering
\begin{subfigure}{2.3in}
  \hspace{-1cm}
  \includegraphics[angle=-90,scale=0.28] {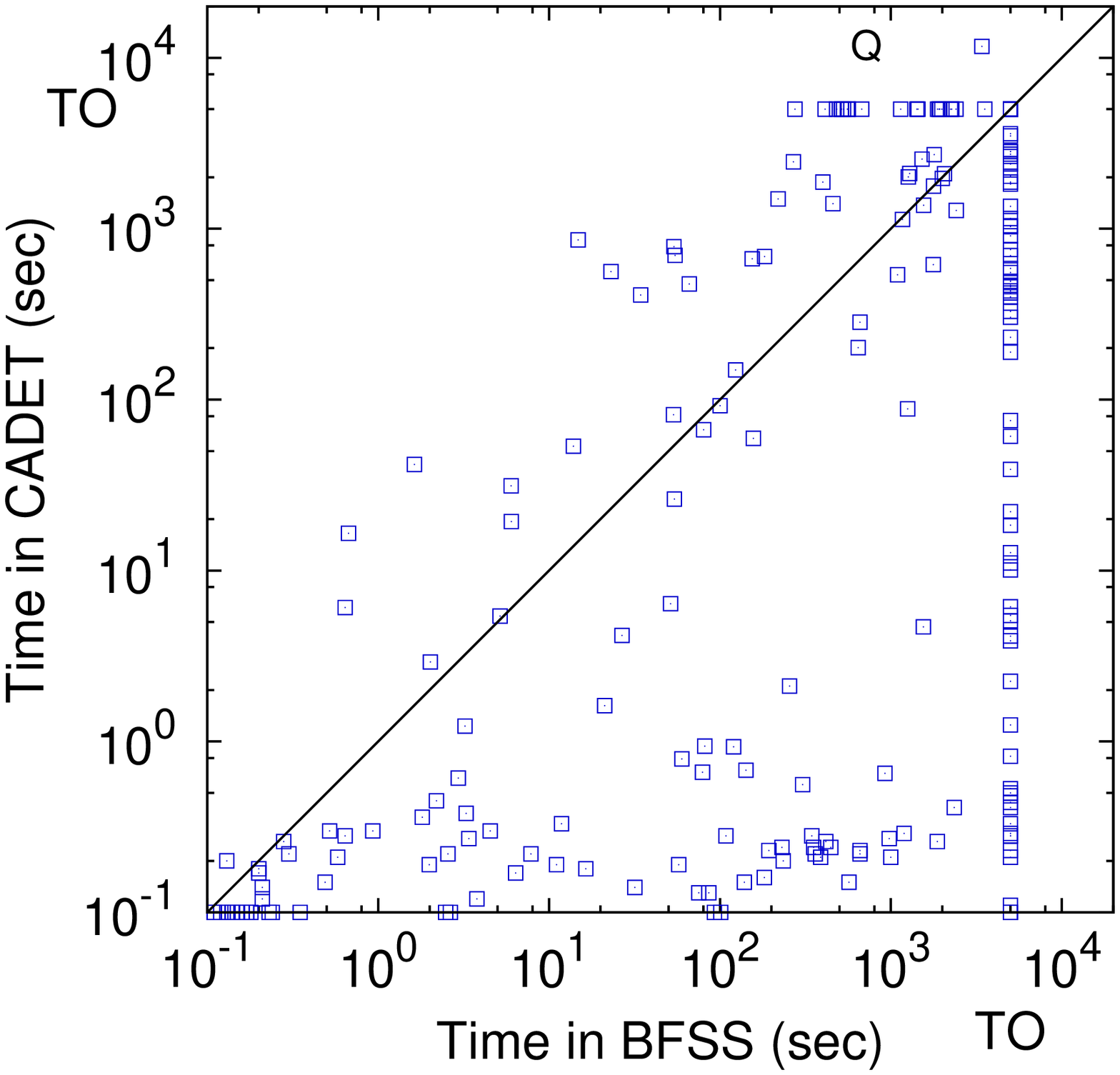} 
\end{subfigure}
\begin{subfigure}{2.3in}
  \includegraphics[angle=-90,scale=0.28] {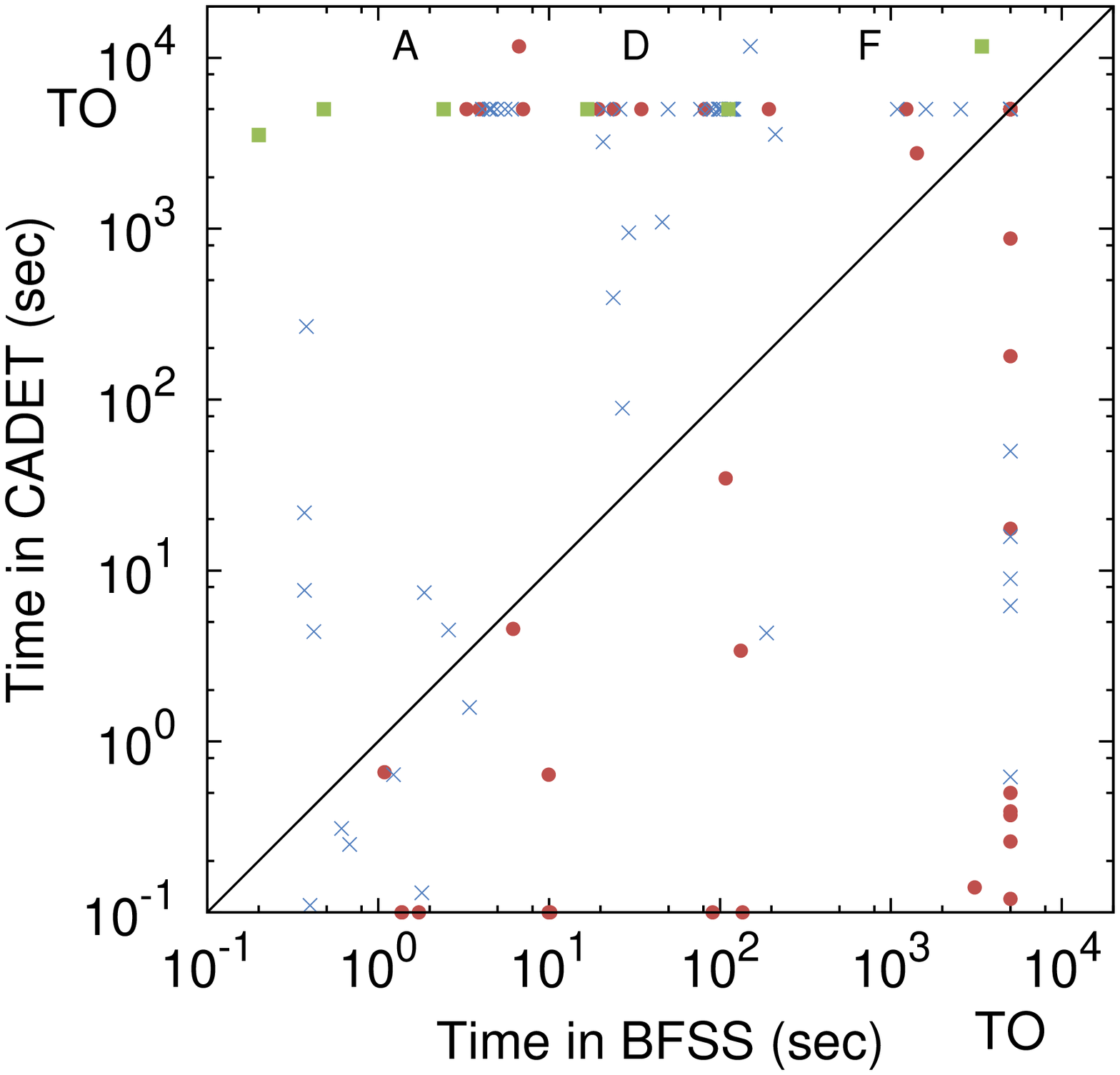} 
\end{subfigure}
\vspace{0.3cm}
\caption{$\bfss$ (BDD-wDNNF  Pipeline) vs $\cadet$ (for legend see Figure \ref{fig:bfsscadetAIG})}
\label{fig:bfsscadetBDD}
\end{figure}

Figure \ref{fig:bfsscadetAIG} gives the  performance of $\bfss$ versus $\cadet$. The performance of $\cadet$ and $\bfss$ is comparable, with $\cadet$ solving $74$ benchmarks across all domains that $\bfss$ could not and $\bfss$ solving $73$ benchmarks that $\cadet$ could not. While $\cadet$ takes less time on many QBFEval benchmarks, on many Arithmetic, Disjunctive Decomposition and Factorization Benchmarks, the BDD-wDNNF pipeline of $\bfss$ takes less time.

\begin{figure}[h]
\centering
\begin{subfigure}{2.3in}
  \hspace{-1cm}
  \includegraphics[angle=-90,scale=0.28] {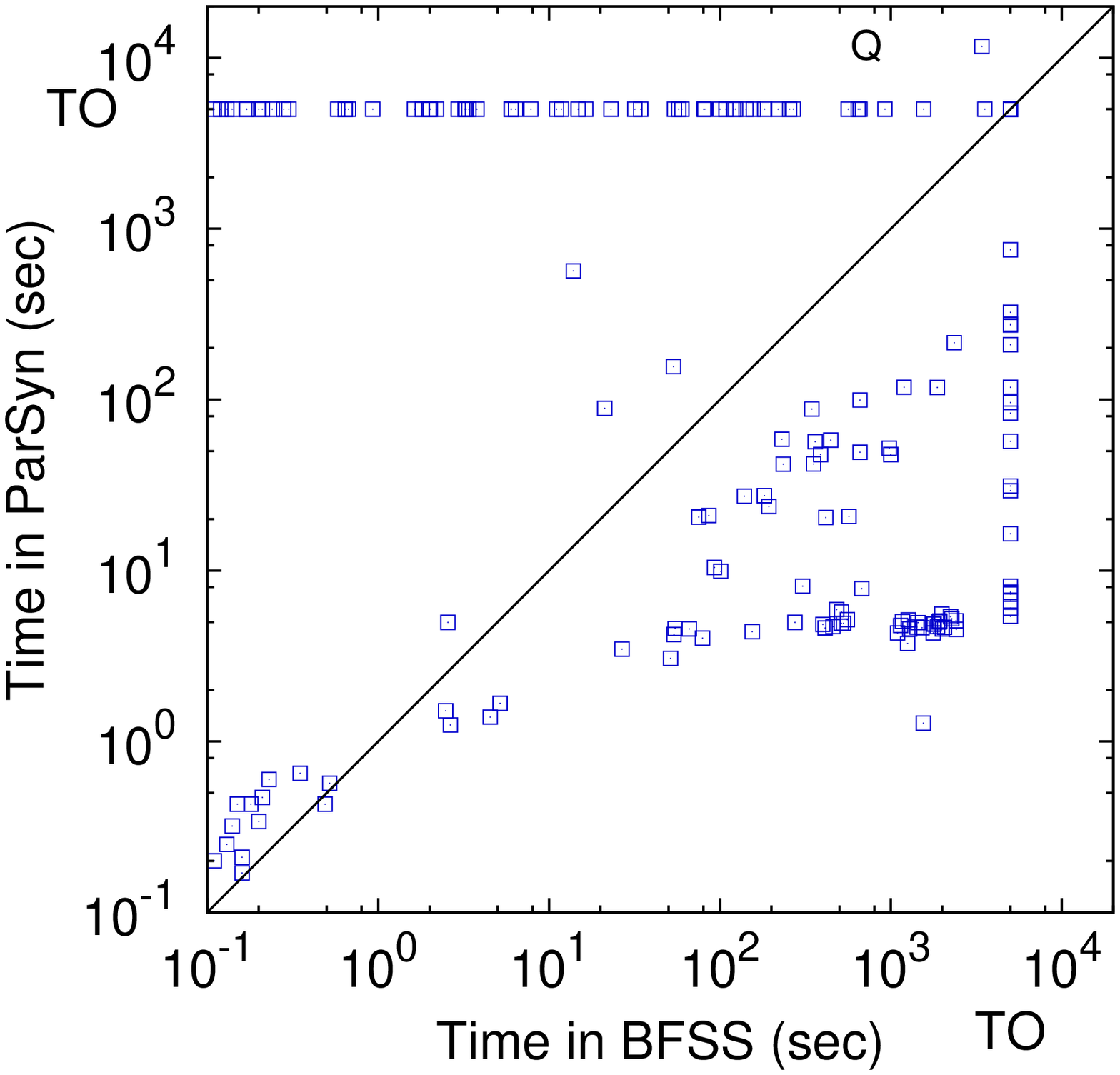} 
\end{subfigure}
\begin{subfigure}{2.3in}
  \includegraphics[angle=-90,scale=0.28] {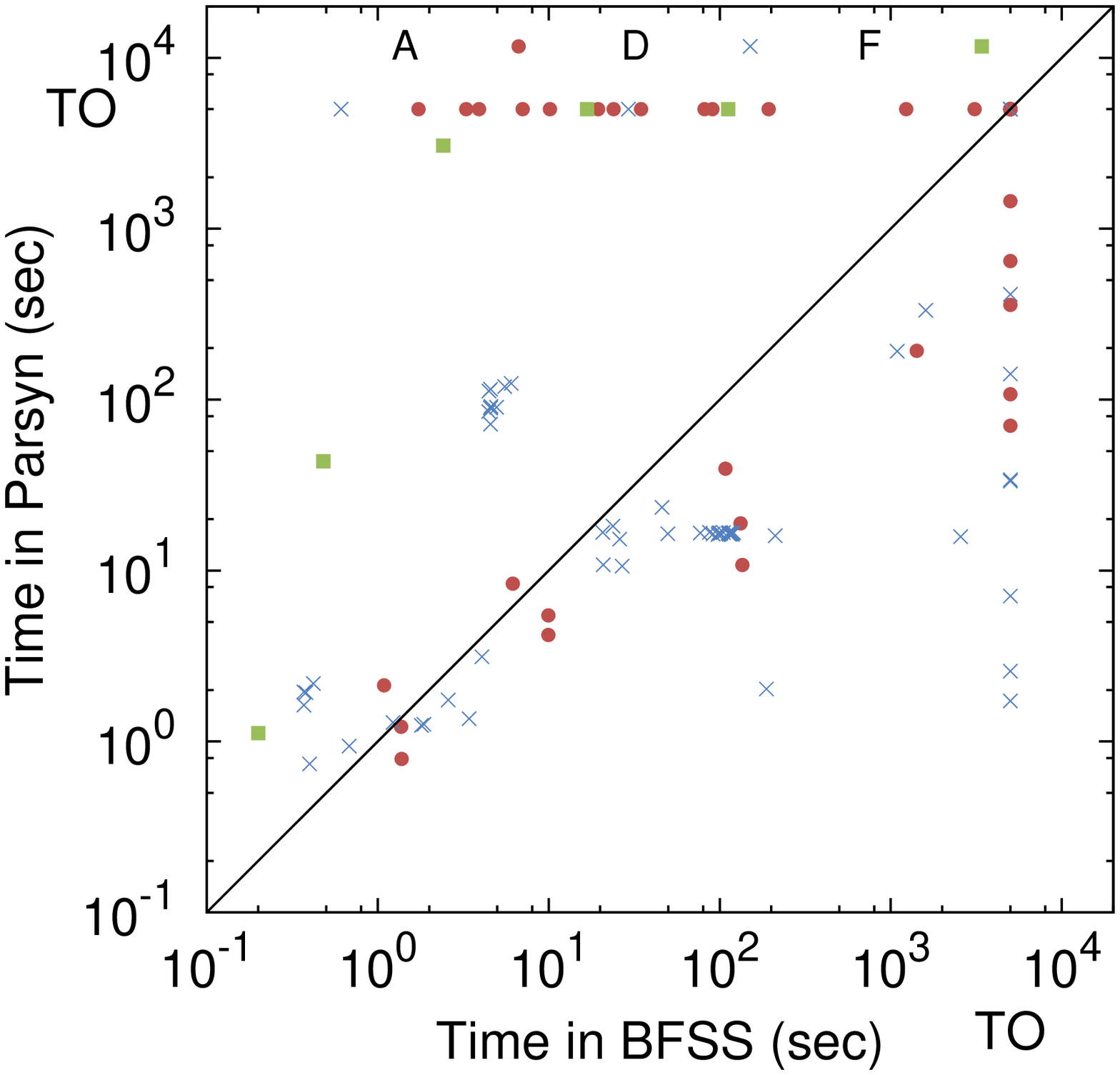} 
\end{subfigure}
\vspace{0.3cm}
\caption{$\bfss$ (BDD-wDNNF  Pipeline)vs $\parsyn$ (for legend see Figure \ref{fig:bfsscadetAIG})}
\label{fig:bfssparsynBDD}
\end{figure}
Figure \ref{fig:bfssparsynBDD} gives the performance of $\bfss$ versus $\parsyn$. While $\parsyn$ could solve $30$ benchmarks across all domains that $\bfss$ could not, the BDD-wDNNF pipeline of $\bfss$ solved $75$ benchmarks that $\parsyn$ could not.

\begin{figure}[h]
\centering
\begin{subfigure}{2.3in}
  \hspace{-1cm}
  \includegraphics[angle=-90,scale=0.28] {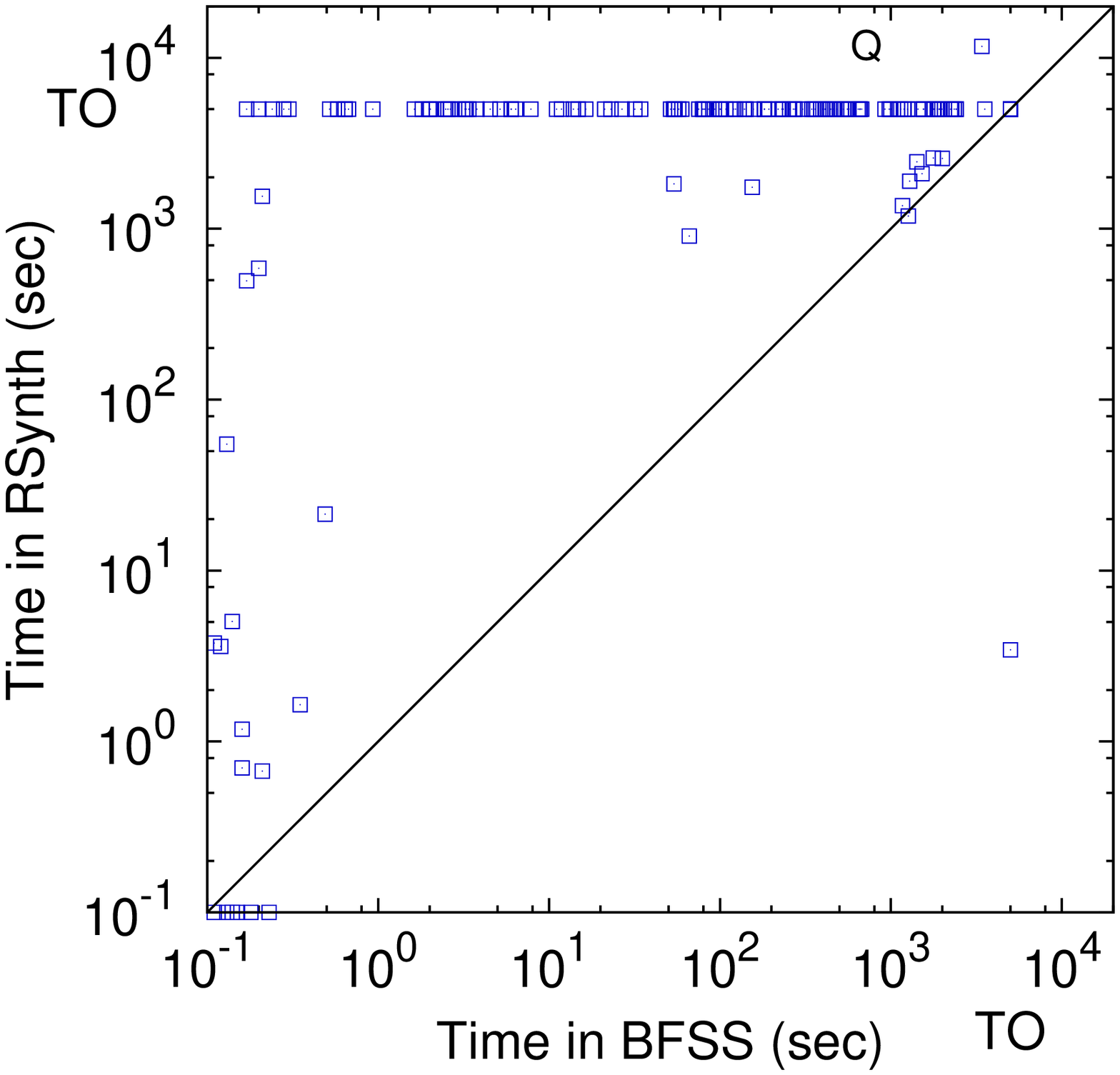} 
\end{subfigure}
\begin{subfigure}{2.3in}
  \includegraphics[angle=-90,scale=0.28] {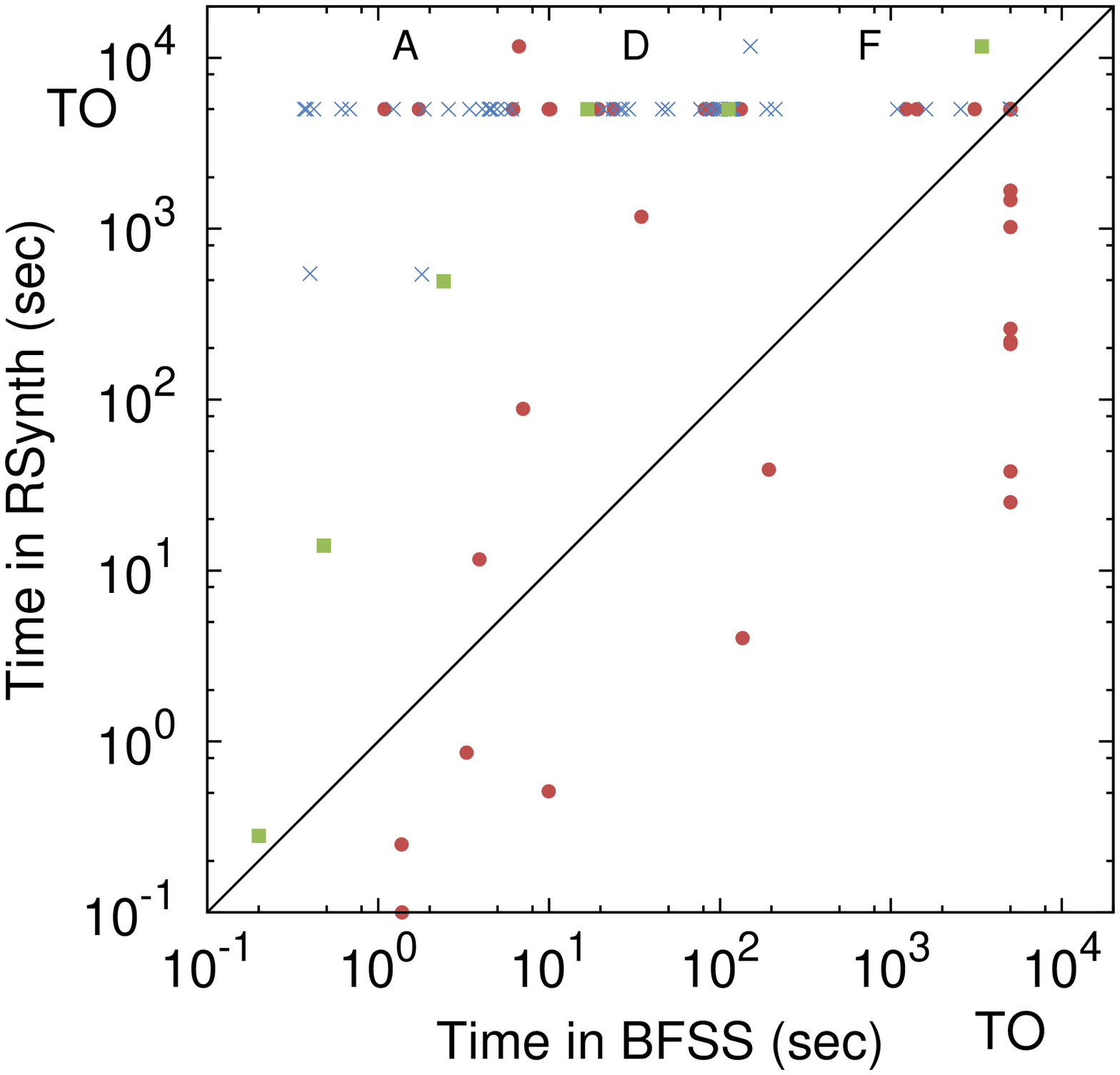} 
\end{subfigure}
\vspace{0.3cm}
\caption{$\bfss$ (BDD-wDNNF  Pipeline) vs $\rsynth$ (for legend see Figure \ref{fig:bfsscadetAIG})} 
\label{fig:bfssrsynthBDD}
\end{figure}
Figure \ref{fig:bfssrsynthBDD} gives the performance of $\bfss$ versus $\rsynth$. While $\rsynth$ could solve $9$ benchmarks across all domains that $\bfss$ could not, the BDD-wDNNF pipeline of $\bfss$ solved $188$ benchmarks that $\rsynth$ could not. Furthermore from Figure \ref{fig:bfssrsynthBDD}, we can see that on most benchmarks, which both the tools could solve, $\bfss$ takes less time.

\begin{figure}[h]
\centering
\begin{subfigure}{2.3in}
  \hspace{-1cm}
  \includegraphics[angle=-90,scale=0.28] {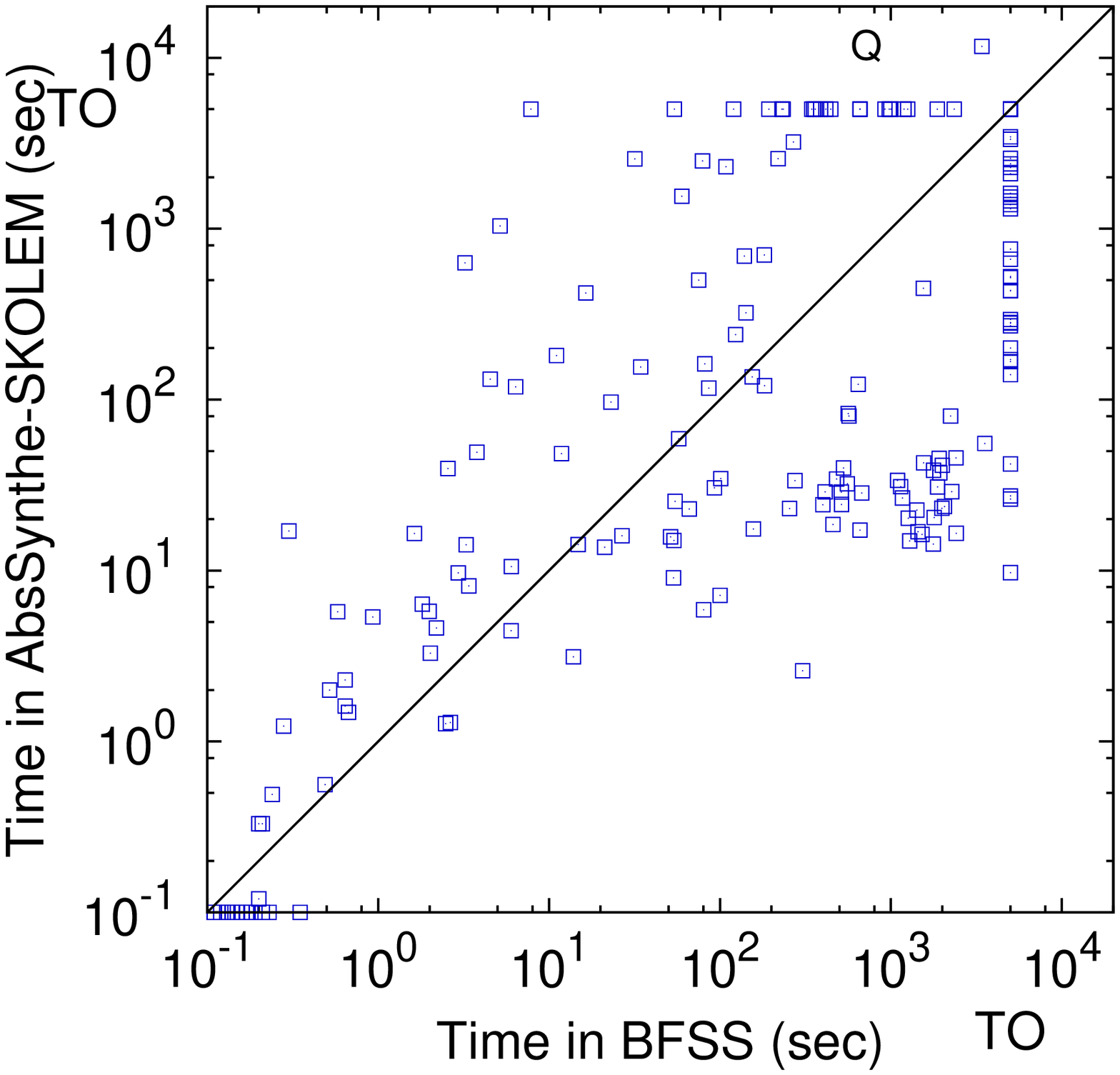} 
\end{subfigure}
\begin{subfigure}{2.3in}
  \includegraphics[angle=-90,scale=0.28] {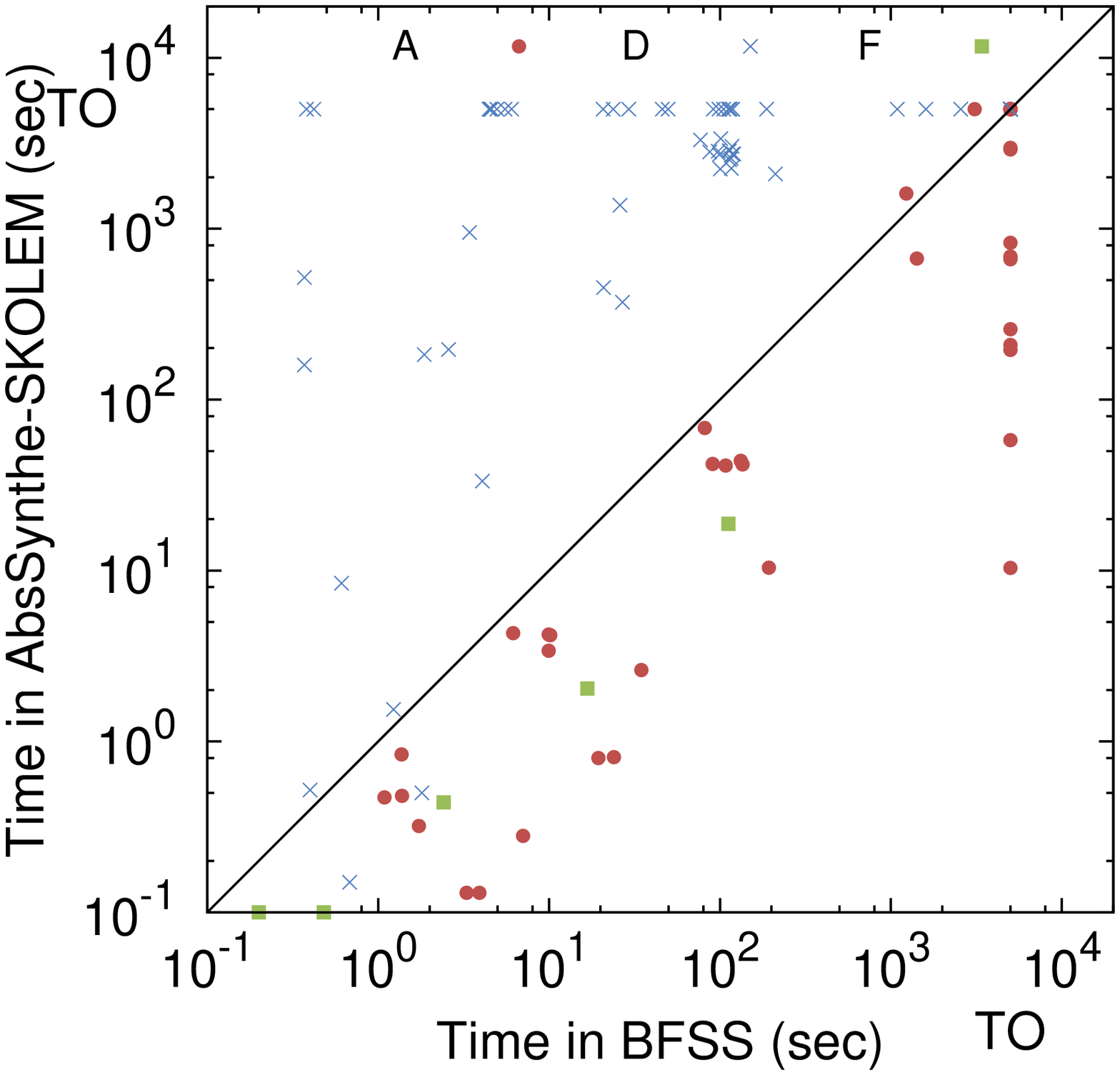} 
\end{subfigure}
\vspace{0.3cm}
\caption{$\bfss$ (BDD-wDNNF  Pipeline) vs $\abssynsk$ (for legend see Figure \ref{fig:bfsscadetAIG})} 
\label{fig:bfssabsBDD}
\end{figure}

Figure \ref{fig:bfssabsBDD} gives the performance of $\bfss$ versus $\abssynsk$. While $\abssynsk$ could solve $39$ benchmarks across all domains that $\bfss$ could not, the BDD-wDNNF pipeline of $\bfss$ solved $52$ benchmarks which could not be solved by $\abssynsk$.

\subsection{Comparison of the two pipelines}

\begin{figure}[h]
\centering
\begin{subfigure}{2.3in}
  \hspace{-1cm}
  \includegraphics[angle=-90,scale=0.28] {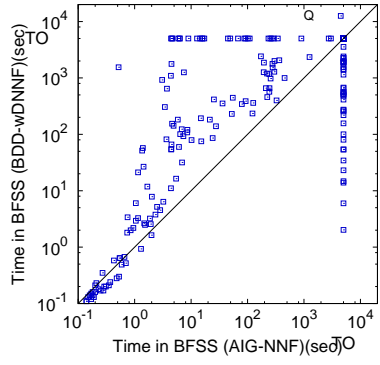} 
\end{subfigure}
\begin{subfigure}{2.3in}
  \includegraphics[angle=-90,scale=0.28] {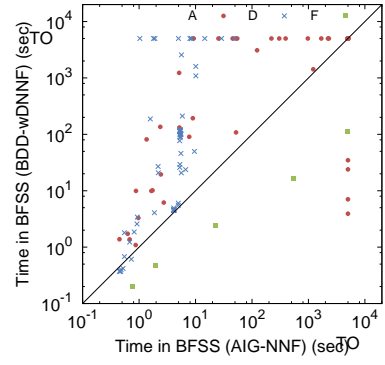} 
\end{subfigure}
\vspace{0.3cm}
\caption{$\bfss$ (AIG-NNF) vs $\bfss$ (BDD-wDNNF) (for legend see Figure \ref{fig:bfsscadetAIG})}
\label{fig:bfssvsbfss}
\end{figure}
Figure \ref{fig:bfssvsbfss} compares the performances of the two pipelines.  
We can see that while there were some benchmarks which only one of the pipelines could solve, apart from Factorization benchmarks, for most of the QBFEval, Arithmetic and Disjunctive Decomposition Benchmarks, the time taken by the AIG-NNF pipeline was less than the time taken by the BDD-wDNNF pipeline.

\end{document}